\documentclass[superscriptaddress, amsmath,amssymb, aps, pra, letterpaper, tightenlines, reprint, notitlepages]{revtex4-2}

\usepackage{graphicx}
\usepackage{floatrow}

\usepackage{dcolumn}
\usepackage{bm}

\usepackage{amsfonts}
\usepackage{dsfont}
\usepackage{bm}
\usepackage{bbold}
\usepackage{xcolor}
\usepackage{lipsum,babel}
\usepackage[normalem]{ulem}

\usepackage[utf8]{inputenc}
\usepackage{amsmath}
\usepackage{physics}
\usepackage{hyperref}
\usepackage{amssymb}
\usepackage{braket}
\usepackage{mathtools}
\usepackage{color}
\usepackage{lineno}

\usepackage{amsthm}

\DeclareMathOperator*{\E}{\mathbb{E}}

\DeclareMathOperator*{\Per}{\text{Per}}

\newcommand{\cor}[1]{{\color{black}{#1}}}

\theoremstyle{plain}
\newtheorem{theorem}{Theorem}
\newtheorem{corollary}{Corollary}

\newtheorem{lemma}{Lemma}

\theoremstyle{definition}
\newtheorem{problem}{Problem}

\begin{document}
	
\definecolor{red}{RGB}{255,0,0}
\preprint{APS/123-QED}


\title{Quantum computational advantage of noisy boson sampling with \\ partially distinguishable photons}
\author{Byeongseon Go}
\affiliation{IRC NextQuantum, Department of Physics and Astronomy, Seoul National University, Seoul 08826, Republic of Korea}
\author{Changhun Oh}
\email{changhun0218@gmail.com}
\affiliation{Department of Physics, Korea Advanced Institute of Science and Technology, Daejeon 34141, Republic of Korea}
\author{Hyunseok Jeong}
\email{h.jeong37@gmail.com}
\affiliation{IRC NextQuantum, Department of Physics and Astronomy, Seoul National University, Seoul 08826, Republic of Korea}

\begin{abstract}
Boson sampling stands out as a promising approach toward experimental demonstration of quantum computational advantage. 
However, the presence of physical noise in near-term experiments hinders the realization of the quantum computational advantage with boson sampling. 
Since physical noise in near-term boson sampling devices is inevitable, precise characterization of the boundary of noise rates where the classical intractability of boson sampling is maintained is crucial for quantum computational advantage using near-term devices.
In this work, we identify the level of partial distinguishability noise that upholds the classical intractability of boson sampling.
We find that boson sampling with on average $O(\log N)$ number of distinguishable photons out of $N$ input photons maintains the equivalent complexity to the ideal boson sampling case. 
By providing strong complexity theoretical evidence for the classical intractability of noisy boson sampling, we expect that our findings will ultimately facilitate the demonstration of quantum computational advantage with noisy boson sampling experiments in the near future.
    
\end{abstract}

\maketitle

\section{Introduction}
The development of noisy intermediate-scale quantum devices~\cite{preskill2018quantum} in recent years has raised expectations for the experimental realization of quantum computational advantage. 
Various computational tasks have been proposed to achieve a quantum computational advantage, i.e., hard to solve with classical computers but efficiently solvable with quantum computers~\cite{feynman2018simulating, lloyd1996universal, shor1999polynomial, bouland2019complexity, aaronson2011computational}. 
Among these tasks, boson sampling~\cite{aaronson2011computational, hamilton2017gaussian, deshpande2022quantum, grier2022complexity} has become
one of the most promising candidates for experimental demonstration of quantum computational advantage due to its experimental feasibility and strong complexity-theoretical evidence of its hardness.
In fact, we have now seen several boson sampling experiments~\cite{zhong2020quantum, zhong2021phase, madsen2022quantum, deng2023gaussian} whose system size is sufficiently large to claim a quantum computational advantage. 

However, despite the significant progress in boson sampling experiments, the degrees of physical noises, such as photon loss and partial distinguishability of photons, in current experimental setups are still large.  
These noises can hinder the experimental demonstration of quantum computational advantage with boson sampling,
because such physical noises can considerably reduce its computational complexity.
More specifically, numerous classical algorithms have been developed to simulate boson sampling subject to physical noises, such as photon loss and partial distinguishability of photons~\cite{kalai2014gaussian, renema2018classical, renema2018efficient, moylett2019classically, shchesnovich2019noise, renema2020simulability, shi2022effect, van2024efficient, oszmaniec2018classical, garcia2019simulating, qi2020regimes, brod2020classical, villalonga2021efficient, bulmer2022boundary, oh2024classicalalgorithm, oh2023classical, oh2024classical, oh2025recent}. 
These classical algorithms can efficiently simulate the noisy boson sampling if the noise level is sufficiently large, possibly ruling out the quantum computational advantage of boson sampling when the noise rate quickly increases with the system size.

Since physical noises are unavoidable in near-term quantum devices and often significantly reduce the computational complexity of the systems, to achieve a quantum computational advantage using near-term quantum devices, it is crucial to identify the boundary of the regime that is classically intractible under the effect of noise, i.e., how large amount of noise is tolerated for the hardness result. 
However, the current theoretical hardness results of boson sampling (as in Refs.~\cite{aaronson2011computational, hamilton2017gaussian, deshpande2022quantum, grier2022complexity, bouland2022noise, bouland2023complexity, bouland2024average}) are generally for the ideal case without any physically realistic noise; thus, these results may not provide the classical hardness evidence of boson sampling when the physical types of noise are applied.

To resolve this issue, it is essential to find the hardness evidence of noisy boson sampling itself.
One promising approach proposed by Ref.~\cite{aaronson2016bosonsampling} is to verify that classical simulation of noisy boson sampling is \textit{as hard as} classical simulation of ideal boson sampling. 
This can be done by showing that the average-case hardness of the ideal boson sampling, a key ingredient for the classical hardness of boson sampling~\cite{aaronson2011computational, hamilton2017gaussian, deshpande2022quantum, grier2022complexity, bouland2022noise, bouland2023complexity, bouland2024average}, can carry over to the average-case hardness of the noisy boson sampling with its noise rate below a certain threshold. 
Then, such noisy boson sampling is equivalently hard to classically simulate as the ideal boson sampling because the average-case hardness of the ideal boson sampling leads to the average-case hardness of the noisy boson sampling.
Ref.~\cite{aaronson2016bosonsampling} has conducted this analysis for the photon loss, which is a dominant noise source in current optical setups.  
More specifically, Ref.~\cite{aaronson2016bosonsampling} has shown that for $N$ single-photon input, noisy boson sampling where at most a fixed $k = O(1)$ number of photon loss (i.e., $N-k$ photons survive) maintains the same complexity as the ideal boson sampling case. 
Since the loss rate in current experiments is obviously much larger than this regime, increasing the noise threshold above the $k = O(1)$ number of photon loss in Ref.~\cite{aaronson2016bosonsampling} remains an important open problem.

Meanwhile, current boson sampling experiments are susceptible to another type of physical noise, partial distinguishability of photons.  
It is one of the major noise sources that hinder the quantum computational advantage of boson sampling~\cite{renema2018efficient, renema2018classical, moylett2019classically, shchesnovich2019noise, renema2020simulability, shi2022effect, van2024efficient}.
More specifically, the indistinguishable nature of photons is key to the computational hardness of boson sampling, and the computational complexity of boson sampling reduces as photons become more distinguishable from each other. 
Also, partial distinguishability is a noise type that appears not only in optical systems~\cite{zhong2020quantum, zhong2021phase, madsen2022quantum, deng2023gaussian} but also across broader bosonic systems like atomic arrays and ion traps; those systems recently served as platforms for boson sampling experiments~\cite{chen2023scalable, young2024atomic} and are expected to be promising candidates to experimentally demonstrate quantum computational advantage.
Therefore, to achieve quantum computational advantage with near-term boson sampling experiments, it is crucial to investigate the classical intractability of noisy boson sampling subject to partial distinguishability noise. 
Here, as we have previously outlined, this can be accomplished by identifying the threshold of the partial distinguishability noise that maintains the classical intractability of the ideal boson sampling.

In this work, we show that noisy $N$-single photon boson sampling with $O(\log N)$ distinguishable photons, on average, still maintains the classical intractability of the ideal boson sampling.
Specifically, we show that the average-case estimation of output probabilities of noisy boson sampling with $O(\log N)$ distinguishable photons is complexity-theoretically equivalent to the ideal boson sampling case.
This indicates that boson sampling with the logarithmic number of distinguishable photons is classically hard under the conjecture that the ideal boson sampling is classically hard. 
We also generalize our result to noisy boson sampling when both photon loss and partial distinguishability are applied and find that at most $O(\log N)$ number of photon loss and $O(\log N)$ number of distinguishable photons still upholds the classical intractability of the ideal boson sampling.

We expect that our findings can help understand the conditions of classical intractability for noisy boson sampling, which has been considerably underexplored compared to the numerous classical simulability arguments of noisy boson sampling~\cite{renema2018classical, renema2018efficient, moylett2019classically, renema2020simulability, shi2022effect, van2024efficient, shchesnovich2019noise, oszmaniec2018classical, garcia2019simulating, qi2020regimes, brod2020classical, villalonga2021efficient, bulmer2022boundary, oh2023classical, oh2024classicalalgorithm, oh2024classical}. 
Further, by providing strong complexity-theoretical evidence of classical simulation hardness of noisy boson sampling, we expect that our results will enable one to fully demonstrate quantum computational advantage with noisy boson sampling experiments in the near future. 

Our paper is organized as follows.
\cor{In Sec.~\ref{Section: related works}, we give a brief overview of the related results to our work.}
In Sec.~\ref{section: output probability}, we establish the noise model for partial distinguishability of photons and introduce the output probability of noisy boson sampling with partial distinguishability. 
In Sec.~\ref{section:problemset}, we formally define our main problem, the average-case estimation of output probabilities of noisy boson sampling with partial distinguishability, and introduce our main result.
In Sec.~\ref{Section:reduction}, we prove that the average-case estimation problem of noisy boson sampling with partial distinguishability is equivalently hard to the ideal boson sampling case.
In Sec.~\ref{section: generalization}, we generalize our main result to the classical intractability of noisy boson sampling when both photon loss and partial distinguishability noise are applied. 
Finally, in Sec.~\ref{section: conclusion}, we provide some concluding remarks and directions for future work.

\cor{
\section{Related works}\label{Section: related works}

Before proceeding, we briefly review related results, and clarify our contributions within the context of the existing literature.

We first go through the existing classical simulability results for partially distinguishable boson sampling.
Reference~\cite{renema2018efficient} presents an efficient approximate classical algorithm for partially distinguishable boson sampling, and finds that it becomes more accurate as the distinguishability of photons increases.
References~\cite{renema2018classical, moylett2019classically} establish the combined noise model for photon loss and partial distinguishability, and similarly propose a classical algorithm to approximately simulate that noisy boson sampling. 
Reference~\cite{shchesnovich2019noise} modifies the Gaussian noise proposed in Ref.~\cite{kalai2014gaussian} to realistic partial distinguishability noise, and suggests a noise threshold for efficient classical simulability of that noisy boson sampling. 
Indeed, these results provide a noise threshold for the classical simulability of partially distinguishable boson sampling, providing a necessary condition on the noise rate to achieve the quantum computational advantage. 
On the other hand, our work complements these analyses by providing a noise threshold for the classical intractability of partially distinguishable boson sampling, establishing a sufficient condition on the noise rate required to achieve quantum computational advantage.

Meanwhile, for the classical hardness result, Ref.~\cite{shchesnovich2014sufficient} shows that partially distinguishable boson sampling whose single-photon mode mismatch is given by $O(N^{-3/2})$ on $N$ input photons, corresponding to $O(N^{-1/2})$ number of distinguishable photons on average, is complexity-theoretically equivalent to the ideal boson sampling.
Further, Ref.~\cite{aaronson2016bosonsampling} establishes complexity-theoretic arguments on lossy boson sampling, deriving that a fixed $O(1)$ number of lost photons on input $N$ photons maintains the classical intractability of ideal boson sampling.
One limitation of these earlier results is that the number of tolerable noisy photons for the classical hardness does not \textit{scale} with the system size.
Unlike these previous results, we show that partially distinguishable boson sampling with $O(\log N)$ number of distinguishable photons on average still maintains its classical hardness. 
Hence, our work significantly improves the noise robustness for the classical hardness by allowing the number of tolerable noisy photons to scale with the system size, thereby paving the way for realizing quantum computational advantage with noisy boson sampling experiments.

For further related works and details, we refer the readers to Appendix~\ref{appendix: section: related works}, which provides a detailed overview of the existing literature on the complexity of noisy boson sampling.

}

\section{Boson sampling with partial distinguishability}\label{section: output probability}

In this section, we introduce the output probability of the noisy boson sampling subject to partial distinguishability noise.
We first introduce the output probability of the ideal boson sampling and later describe how it varies when partial distinguishability noise is introduced. 

We consider the standard Fock-state boson sampling scheme in Ref.~\cite{aaronson2011computational} with the following setup:
We prepare $N$ single-photon states and inject them into the first $N$ modes over an $M$-mode linear optical network, characterized by an $M$ by $M$ Haar-random unitary matrix $U$.
Here, $M$ is polynomially related to $N$ as $M = O(N^{\gamma})$ for a large constant $\gamma$, as depicted in Ref.~\cite{aaronson2011computational}. 
After the unitary evolution, we measure the output photon number for each mode.
Here, the measurement outcome can be represented as an $N$-dimensional integer vector $S = (s_1,\dots,s_N)$ with $1\leq s_i\leq M$ and $s_1\le s_2\le\cdots\le s_N$, where each element represents the output mode in which a photon is detected.
Then, the output probability $p_{S}(U)$ to obtain the outcome $S$ can be expressed as~\cite{aaronson2011computational} 
\begin{equation}
    p_{S}(U) = \frac{1}{\mu(S)}|\Per(U_{S,1})|^2,
\end{equation}
where $\mu(S)$ denotes a product of the multiplicity of all possible values in $S$ (e.g., $\mu(S) = 1$ for collision-free outcome $S$), and $U_{S,1}$ is an $N$ by $N$ matrix defined by taking rows of $U$ according to $S$, and taking the first $N$ columns of $U$.

By the hiding property shown in Ref.~\cite{aaronson2011computational}, given $M$ by $M$ Haar-random unitary matrix $U$ with a sufficiently large $\gamma$, $U_{S,1}$ is close to $X/\sqrt{M}$ for a complex i.i.d. Gaussian random matrix $X \sim \mathcal{N}(0,1)_{\mathbb{C}}^{N\times N}$, for any collision-free outcome $S$.
Therefore, for Haar-random unitary matrix $U$ and for any collision-free outcome $S$, the output probability $p_{S}(U)$ can also be represented as $q(X)/M^{N}$ for $X \sim \mathcal{N}(0,1)_{\mathbb{C}}^{N\times N}$, where
\begin{equation}\label{idealoutputprobability}
    q(X) = |\Per(X)|^2 = \sum_{\sigma,\rho\in S_N}\prod_{i=1}^{N}X_{\sigma(i),i}X_{\rho(i),i}^* ,
\end{equation}
is the rescaled version of the output probability $p_{S}(U)$ of the ideal boson sampling multiplied by $M^N$.


\cor{
We now describe how $q(X)$ in Eq.~\eqref{idealoutputprobability} changes when partial distinguishability noise is applied.
For the partial distinguishability noise, we employ the noise model of \textit{mutual indistinguishability} between photon pairs, depicted in Refs.~\cite{tichy2015sampling, renema2018efficient, moylett2019classically}. 
More precisely, the mutual indistinguishability for $N$ photons can be quantified by an $N$ by $N$ Gram matrix $S_{ij} = \braket{\Psi_{i}|\Psi_{j}}$ for $\Psi_{i}$ being the $i$-th single photon wave function.
We note that this mutual indistinguishability between photon pairs (i.e., the elements of the Gram matrix) can be experimentally inferred via Hong-Ou-Mandel interference~\cite{hong1987measurement}, as demonstrated in the recent boson sampling experiment~\cite{young2024atomic}.
Specifically, this can be done by experimentally obtaining the coincidence probability of a pair of photons and comparing it with that of a perfectly distinguishable pair of photons~\cite{young2024atomic}.

Throughout this work, similarly to Refs.~\cite{tichy2015sampling, renema2018efficient, moylett2019classically}, we consider the simplified Gram matrix $S_{ij} = x + (1-x)\delta_{ij}$ for $x \in [0,1]$, which represents the \textit{uniform} mutual indistinguishability $x$ among all photon pairs.
Here and throughout, we refer to $x$ as the indistinguishability rate of boson sampling.}
Using this noise model, photons become more distinguishable as $x$ decreases;
$x=1$ corresponds to the fully indistinguishable case, $0<x<1$ corresponds to a partially distinguishable case, and $x=0$ corresponds to the fully distinguishable case.
After the indistinguishability $x$ is introduced, the (rescaled) noisy output probability $q(x, X)$ transforms from $q(X)$ as~\cite{tichy2015sampling, renema2018efficient}
\begin{equation}\label{outputprobabilitywhendistinguishabilityapplied}
    q(x, X) = \sum_{\sigma,\rho\in S_N}x^{N-(\sigma\cdot\rho)} \prod_{i=1}^{N}X_{\sigma(i),i}X_{\rho(i),i}^*,
\end{equation}
where $(\sigma\cdot\rho)$ is the number of $i$'s such that $\sigma(i) = \rho(i)$ for all $i \in [N]$.
One can readily check the two extremal cases: 
For $x = 1$, $q(x, X) = |\Per(X)|^2$, which reproduces a (rescaled) output probability of the ideal boson sampling and known to be \#P-hard to calculate to within a multiplicative error~\cite{aaronson2011computational}.
On the other hand, for $x=0$, $q(x, X) = \Per(|X|^2)$, which can be calculated to within a multiplicative error in polynomial time due to the nonnegativity of the matrix elements~\cite{jerrum2004polynomial}. 
Clearly, this clarifies a tendency that the lower the indistinguishability rate $x$ (i.e., photons become more distinguishable from each other), the easier it is to compute the output probability.

More importantly, as proposed in Ref.~\cite{moylett2019classically}, the indistinguishability rate $x$ in Eq.~\eqref{outputprobabilitywhendistinguishabilityapplied} can also be interpreted as the probability that an individual photon remains indistinguishable (in other words, each photon becomes distinguishable with probability $1 - x$).
Accordingly, the output probability $q(x, X)$ in Eq.~\eqref{outputprobabilitywhendistinguishabilityapplied} can also be represented by a linear sum of output probabilities for all possible indistinguishable photon number $j \in \{0,1,\dots,N\}$, where their coefficients are given by the binomial distribution of $N$ and $x$ as 
\begin{align}\label{noisyoutputprobability}
\begin{split}
    &q(x, X) \\
    &= \sum_{j=0}^{N}x^{j}(1-x)^{N-j}  \sum_{\substack{I \subseteq [N]\\|I|=j}} \sum_{\substack{J \subseteq [N]\\|J|=j}} |\Per(X_{I,J})|^2 \Per(|X_{\bar{I},\bar{J}}|^2)  \\
    &=  \sum_{j=0}^{N}\binom{N}{j}x^{j}(1-x)^{N-j}q_{j}(X),
\end{split}
\end{align}
where $I$ and $J$ represent all the possible subsets of $[N]$ with their size $j$, and $\bar{I}$ and $\bar{J}$ represent $[N]\setminus I$ and $[N]\setminus J$ each with their size $N - j$.
Here, $X_{I,J}$ is a matrix defined by taking rows and columns of $X$ according to $I$ and $J$ (also similarly for $X_{\bar{I},\bar{J}}$), and 
\begin{align}\label{fixedpdoutputprobability}
    q_{j}(X) = \binom{N}{j}^{-1}\sum_{\substack{I \subseteq [N]\\|I|=j}} \sum_{\substack{J \subseteq [N]\\|J|=j}} |\Per(X_{I,J})|^2 \Per(|X_{\bar{I},\bar{J}}|^2)
\end{align}
denotes the (rescaled) output probability of boson sampling with a fixed number $j \in \{0,1,\dots,N\}$ of indistinguishable photons over $N$ photons. 
For example, $j = N$ corresponds to the output probability with fully indistinguishable photons, and correspondingly, $q_{N}(X) = |\Per (X)|^2$ is the ideal output probability of boson sampling.
Hence, the noisy output probability $q(x, X)$ in Eq.~\eqref{noisyoutputprobability} $contains$ the ideal output probability $q_{N}(X)$ of boson sampling in the summation. 
This later becomes a crucial property for showing the classical intractability of partially distinguishable boson sampling.

\cor{

\section{Our main problem and result}\label{section:problemset}

}

In this section, we introduce our main problem for the classical simulation hardness of partially distinguishable boson sampling, which is a modified version of the main problem addressed in the original boson sampling proposal~\cite{aaronson2011computational} (which we referred to as ideal boson sampling). 
More specifically, proving the classical intractability of the ideal boson sampling relies on \#P-hardness of average-case estimation of the (rescaled) output probability $q(X)$ in Eq.~\eqref{idealoutputprobability} over $X \sim \mathcal{N}(0,1)_{\mathbb{C}}^{N\times N}$, which can be formally represented as follows:

\begin{problem}[$|\text{GPE}|_{\pm}^{2}$~\cite{aaronson2011computational}]
Given as input a matrix $X \sim \mathcal{N}(0,1)_{\mathbb{C}}^{N\times N}$ of i.i.d. Gaussians, together with error bounds $\epsilon_0,\,\delta_0 > 0$, estimate $q(X)$ to within additive error $\pm \epsilon_0\cdot N!$ with probability at least $1-\delta_0$ over $X$ in $\text{poly}(N,\epsilon_0^{-1},\delta_0^{-1})$ time. 
\end{problem}

By Ref.~\cite{aaronson2011computational}, the above $|\text{GPE}|_{\pm}^{2}$ problem is \#P-hard under plausible conjectures. 
Assuming that $|\text{GPE}|_{\pm}^{2}$ is \#P-hard, one can show that simulating the ideal boson sampling within an inverse-polynomial total variation distance error is classically intractable.
More specifically, suppose there exists a randomized classical sampler $\mathcal{S}_{0}$ that simulates the ideal boson sampling: on input a circuit unitary matrix $U$, samples the outcome from the output distribution of boson sampling within total variation distance $\beta$. 
Given that the input $U$ is an $M$ by $M$ Haar-random unitary matrix, the (rescaled) output probability of boson sampling can be represented as $q(X)$ for $X \sim \mathcal{N}(0,1)_{\mathbb{C}}^{N\times N}$.
Then, by Stockmeyer's reduction~\cite{stockmeyer1985approximation}, $|\text{GPE}|_{\pm}^{2}$ can be solved for $\beta = O(\epsilon_0\delta_0)$, and in this case $|\text{GPE}|_{\pm}^{2} \in \rm{BPP}^{NP^{\mathcal{S}_{0}}}$ (see Ref.~\cite{aaronson2011computational} for more details). 
This implies that if an efficient algorithm for $\mathcal{S}_{0}$ with $\beta = O(\epsilon_0\delta_0)$ exists, $|\text{GPE}|_{\pm}^{2}$ can be solved in $\text{BPP}^{\text{NP}}$, then the polynomial hierarchy collapses to the level $\text{BPP}^{\text{NP}}$ given that $|\text{GPE}|_{\pm}^{2}$ is \#P-hard. 
In other words, given that the polynomial hierarchy does not collapse to the finite level, there cannot be any classical algorithm that can simulate ideal boson sampling within the total variation distance $\beta$ in $\text{poly}(N,\beta^{-1})$ time.

\cor{
Based on this understanding, we now extend this analysis to the noisy version as follows.
First, we denote $x^{*}$ as the \textit{fixed} indistinguishability rate of a partially distinguishable boson sampling hereafter, such that $1 - x^{*}$ characterizes the noise rate of the noisy boson sampling.
To derive the classical hardness result for this noisy boson sampling with indistinguishability rate $x^{*}$, which is the goal of this work, we now introduce our key strategy and assumption: 
we consider an indistinguishability rate $x$ below $x^{*}$ (i.e., $x \in [0, x^{*}]$) as an \textit{input variable} to the classical sampler simulating the noisy boson sampling, and derive the hardness result for this classical sampler throughout our work.
This constitutes the main distinction from the previous hardness results as in Ref.~\cite{aaronson2016bosonsampling}, and represents a crucial aspect of our approach, because, as it turns out, we utilize different values of $x$ to derive our hardness results.

In fact, considering $x \in [0, x^{*}]$ as an input variable is motivated by the crucial observation that the noisy boson sampling becomes classically \textit{easier} as $x$ decreases~\cite{renema2018efficient, renema2018classical, moylett2019classically, shchesnovich2019noise, renema2020simulability, shi2022effect, van2024efficient}.
That is, if noisy boson sampling with indistinguishability $x^{*}$ is classically simulable, then so is noisy boson sampling with indistinguishability $x \leq x^{*}$.
Let us assume that this observation holds generally for any $x^{*}$ and $x \in [0, x^{*}]$, which is a plausible assumption supported by previous findings~\cite{renema2018efficient, renema2018classical, moylett2019classically, shchesnovich2019noise, renema2020simulability, shi2022effect, van2024efficient}.
Then, one can deduce that, the classical hardness of simulating noisy boson sampling that can choose any indistinguishability $x \in [0, x^{*}]$ implies the classical hardness of simulating noisy boson sampling with a fixed indistinguishability $x^{*}$.



}

\cor{To derive the classical hardness of simulating noisy boson sampling that takes an indistinguishability rate $x \in [0, x^{*}]$ as an input, }
we now set the problem for partially distinguishable boson sampling in correspondence to $|\text{GPE}|_{\pm}^{2}$, i.e., average-case estimation of the noisy output probability $q(x, X)$ in Eq.~\eqref{noisyoutputprobability}.
We formally state our main problem: Estimating the sum of Partially-positive Gaussian Permanent, which we will denote as $|\text{PGPE}|_{\pm}^{2}$ hereafter.

\begin{problem}[$|\text{PGPE}|_{\pm}^{2}$]
Given as input an indistinguishability rate \cor{$x \in [0, x^{*}]$ for a fixed $x^{*}$} and a matrix $X \sim \mathcal{N}(0,1)_{\mathbb{C}}^{N\times N}$ of i.i.d. Gaussians, together with error bounds $\epsilon,\,\delta > 0$, estimate $q(x, X)$ to within additive error $\pm \epsilon\cdot N!$ with probability at least $1-\delta$ over $X$ in $\text{poly}(N,\epsilon^{-1},\delta^{-1})$ time. 
\end{problem}

\cor{
To highlight the changes in $|\text{PGPE}|_{\pm}^{2}$ from the ideal case $|\text{GPE}|_{\pm}^{2}$, we modify the output probability to account for noise (i.e., $q(X) \rightarrow q(x, X)$), and also introduce an additional input indistinguishability rate $x \in [0, x^{*}]$.
}

\cor{In correspondence to $|\text{PGPE}|_{\pm}^{2}$ problem,} we can formally define a randomized classical sampler $\mathcal{S}$ that simulates a partially distinguishable boson sampling within total variation distance error $\beta$,
where now the inputs of $\mathcal{S}$ are $(i)$ a circuit unitary matrix $U$, and $(ii)$ \cor{an indistinguishability rate $x \in [0, x^{*}]$}.
Then, similarly as before, to show that the classical sampler $\mathcal{S}$ for partially distinguishable boson sampling is intractable, it suffices to prove that $|\text{PGPE}|_{\pm}^{2}$ is \#P-hard.
Specifically, given that the input $U$ is a Haar-random unitary matrix, the (rescaled) output probability of the noisy boson sampling can be represented as $q(x, X)$ for $X \sim \mathcal{N}(0,1)_{\mathbb{C}}^{N\times N}$.
Then, given access to the classical sampler $\mathcal{S}$, on input a Haar random unitary matrix $U$ and \cor{an indistinguishability rate $x \in [0, x^{*}]$}, $|\text{PGPE}|_{\pm}^{2}$ is in complexity class $\rm{BPP}^{NP^{\mathcal{S}}}$ for $\beta = O(\epsilon\delta)$ by Stockmeyer's reduction~\cite{aaronson2011computational, stockmeyer1985approximation}.
Therefore, if $|\text{PGPE}|_{\pm}^{2}$ is \#P-hard, there cannot be any efficient classical algorithm for $\mathcal{S}$ unless the polynomial hierarchy collapses.
\cor{Moreover, under our key assumption that classical simulability does not degrade with decreasing $x$, hardness of $\mathcal{S}$ implies the classical hardness of simulating partially distinguishable boson sampling with a fixed indistinguishability rate $x^{*}$.
}

\cor{

Next, to derive the \#P-hardness of $|\text{PGPE}|_{\pm}^{2}$, we take a similar approach to Ref.~\cite{aaronson2016bosonsampling}, i.e., establishing reduction from $|\text{GPE}|_{\pm}^{2}$ to $|\text{PGPE}|_{\pm}^{2}$. 
Specifically, we aim to identify the minimum (threshold) $x^{*}$ that allows \textit{polynomial} reduction from $|\text{GPE}|_{\pm}^{2}$ to $|\text{PGPE}|_{\pm}^{2}$.
Then, summarizing all the previous arguments, this implies that for $x^{*}$ above such threshold value, partially distinguishable boson sampling with its indistinguishability rate $x^{*}$ is equivalently hard to classically simulate as the ideal boson sampling.}

To quantify the threshold of $x^{*}$ that allows the polynomial reduction from $|\text{GPE}|_{\pm}^{2}$ to $|\text{PGPE}|_{\pm}^{2}$, we employ a convention for the input indistinguishability rate as $x = 1 - \frac{k}{N}$.
Here, because $x$ can be interpreted as the probability of being indistinguishable for each photon~\cite{moylett2019classically}, the average number of indistinguishable photons is ${Nx = N - k}$.
Then, similarly to the photon-loss case as in Ref.~\cite{oszmaniec2018classical}, $k = (1-x)N$ represents the average number of photons subject to noise, and in our case, it represents the average number of distinguishable photons. 
\cor{Also, the fixed indistinguishability rate can be represented as $x^{*} = 1 - \frac{k_{\text{min}}}{N}$, where the subscript min comes from the fact that $x^{*}$ is indeed a maximum value over input indistinguishability rate $x$. }

Finally, we introduce our main result.
We show that for $k_{\text{min}} = O(\log N)$, $|\text{PGPE}|_{\pm}^{2}$ is polynomially reducible from $|\text{GPE}|_{\pm}^{2}$.
Informally, our main result can be expressed as follows.

\begin{theorem}[Informal]\label{theorem1informal}
\cor{For $x^{*}$ satisfying $(1-x^{*})N = O(\log N)$, $|\rm{PGPE}|_{\pm}^{2}$ is polynomially reducible from $|\rm{GPE}|_{\pm}^{2}$.
In other words,} if $\mathcal{O}$ is an oracle that solves $|\rm{PGPE}|_{\pm}^{2}$ for $x^{*} = 1 - \frac{k_{\rm{min}}}{N}$ with $k_{\rm{min}} = O(\log N)$, together with $\epsilon, \delta = O({\rm{poly}}(N, \epsilon_0^{-1}, \delta_0^{-1})^{-1})$, then $|\rm{GPE}|_{\pm}^{2}$ can be solved in $\rm{BPP}^{\mathcal{O}}$. 
\end{theorem}

We provide in Sec.~\ref{Section:reduction} a formal statement of Theorem~\ref{theorem1informal}, followed by a detailed proof of the theorem.
\cor{Here, let us emphasize that, by Theorem~\ref{theorem1informal}, partially distinguishable boson sampling that has on average $(1-x^{*})N = O(\log N)$ number of distinguishable photons is complexity-theoretically equivalent to the ideal boson sampling.}
More specifically, $\epsilon^{-1}$ and $\delta^{-1}$ in Theorem~\ref{theorem1informal} are polynomially related to $N$, $\epsilon_0^{-1}$ and $\delta_0^{-1}$.
Hence, for $\epsilon$ and $\delta$ given in Theorem~\ref{theorem1informal}, the allowed operation time to solve $|\rm{PGPE}|_{\pm}^{2}$ problem is still in $\text{poly}(N,\epsilon_0^{-1},\delta_0^{-1})$. 
Accordingly, by Theorem~\ref{theorem1informal}, if one can classically simulate a noisy boson sampling with up to $k_{\text{min}} = c_{\text{min}}\log N$ number of distinguishable photons on average, to within $O(\epsilon\delta)$ total variation distance, then one can also solve $|\rm{GPE}|_{\pm}^{2}$ problem in $\text{BPP}^{\text{NP}}$. 
Therefore, under the complexity-theoretic conjecture that $|\rm{GPE}|_{\pm}^{2}$ is \#P-hard, we can deduce that such a noisy boson sampling is hard to classically simulate. 

In addition, by using a similar process to the proof of Theorem~\ref{logpdhardness}, we also find that if the average number of distinguishable photons $k$ scales sub-logarithmically with $N$, then the requirement of the approximation error $\epsilon$ for the reduction from $|\rm{GPE}|_{\pm}^{2}$ to $|\rm{PGPE}|_{\pm}^{2}$ is comparably less stringent.

\begin{corollary}\label{sublogpdhardness}
If $\mathcal{O}$ is an oracle that solves $|\rm{PGPE}|_{\pm}^{2}$ for $x^{*} = 1 - \frac{k_{\rm{min}}}{N}$ with $k_{\rm{min}} = o(\log N)$, $\epsilon = O(\epsilon_0^{\alpha+1}\delta_0^{\alpha}N^{-\alpha})$ with any constant $\alpha>0$, and  $\delta =  O\left(\frac{\delta_0}{\log N + \log\epsilon_0^{-1}\delta_0^{-1}}\right)$, then $|\rm{GPE}|_{\pm}^{2}$ can be solved in $\rm{BPP}^{\mathcal{O}}$. 
\end{corollary}

\begin{proof}
    See Appendix~\ref{appendix:proofofcorollary}.
\end{proof}


\section{Reducing $|\text{GPE}|_{\pm}^{2}$ to $|\text{PGPE}|_{\pm}^{2}$ for $x^{*} = 1 - \frac{k_{\rm{min}}}{N}$ with $k_{\text{min}} = O(\log N)$}\label{Section:reduction}

In this section, we establish a reduction from $|\text{GPE}|_{\pm}^{2}$ to $|\text{PGPE}|_{\pm}^{2}$, especially for the case that the minimum average number of distinguishable photons scales logarithmically: $k_{\text{min}} = O(\log N)$. 
Here, to simplify the analysis, we set $k_{\text{min}}$ as $c_{\text{min}}\log N$ for a constant $c_{\text{min}}$, without loss of generality.
Using this convention, we restate our main theorem for the reduction from $|\text{GPE}|_{\pm}^{2}$ to $|\text{PGPE}|_{\pm}^{2}$, which can be formally represented as follows.

\begin{theorem}\label{logpdhardness}
If $\mathcal{O}$ is an oracle that solves $|\rm{PGPE}|_{\pm}^{2}$ for $x^{*} = 1 - \frac{k_{\rm{min}}}{N}$ with $k_{\rm{min}} = c_{\rm{min}}\log N$, together with $\epsilon = O\left(\epsilon_0^{7.1}\delta_0^{6.1} N^{-40 c_{\text{min}}}\right)$, and $\delta =  O\left(\frac{\delta_0}{\log N + \log\epsilon_0^{-1}\delta_0^{-1}}\right)$, then $|\rm{GPE}|_{\pm}^{2}$ can be solved in $\rm{BPP}^{\mathcal{O}}$. 
\end{theorem}

Because $c_{\text{min}}$ is a constant, $\epsilon^{-1}$ and $\delta^{-1}$ in Theorem~\ref{logpdhardness} are at most polynomially related to $N$, $\epsilon_0^{-1}$ and $\delta_0^{-1}$.
As we have discussed in the previous section, this indicates the complexity-theoretical equivalence of classically simulating partially distinguishable boson sampling with $O(\log N)$ number of distinguishable photons and classically simulating ideal boson sampling.

In the rest of the section, we provide a step-by-step proof of Theorem~\ref{logpdhardness}.

\begin{figure*}[t]
\includegraphics[width=0.9\linewidth]{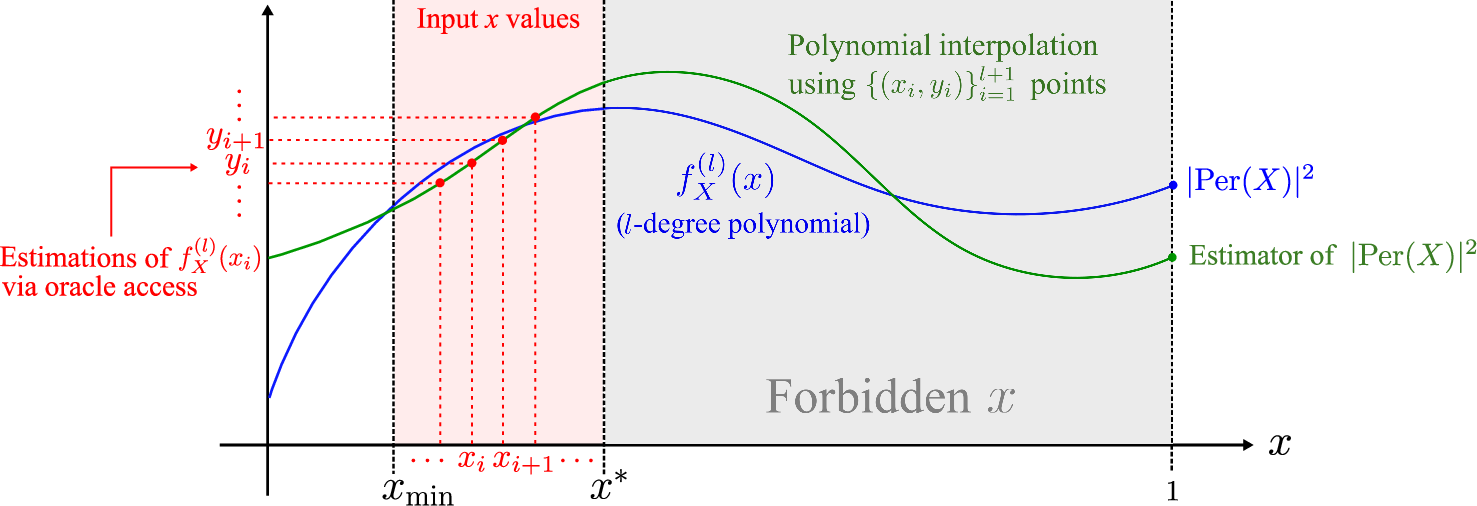}
\caption{Schematics for the proof of Theorem~\ref{logpdhardness}, i.e., how to solve $|\text{GPE}|_{\pm}^{2}$ problem via oracle access to $|\text{PGPE}|_{\pm}^{2}$ problem. 
We first find the $l$-degree polynomial $f_X^{(l)}(x)$ that can be well estimated via oracle access to the $|\text{PGPE}|_{\pm}^{2}$ problem for $x\le x^{*}$, and satisfies $f_X^{(l)}(1) = |\Per(X)|^2$ (blue line).
Then, for each of $l+1$ number of input indistinguishability rate $x_i \in [x_{\text{min}},x^{*}]$, we obtain $y_i$ value that estimate $f_X^{(l)}(x_i)$ via the oracle access (red dots). 
Lastly, we conduct polynomial interpolation with given $\{x_i,y_i\}_{i=1}^{l+1}$ points to infer the polynomial $f_X^{(l)}(x)$ (green line) and finally obtain the estimation value of $|\Per(X)|^2$ at $x=1$. 
}
\label{fig:proofoutline}
\end{figure*}

\subsection{Proof sketch}
Let us first sketch the proof of Theorem~\ref{logpdhardness} by briefly describing the reduction process from $|\text{GPE}|_{\pm}^{2}$ to $|\text{PGPE}|_{\pm}^{2}$.
As the noisy output probability $q(x, X)$ in Eq.~\eqref{noisyoutputprobability} is a polynomial in $x$, given access to $q(x, X)$ values for different values of input $x \le x^{*}$, we use polynomial interpolation technique to infer the value at $x= 1$, i.e., $q(x = 1,X) = |\Per(X)|^2$. 
Indeed, this process is similar to one of the main techniques of error mitigation, i.e., zero-noise extrapolation~\cite{cai2023quantum, temme2017error, li2017efficient}. 
However, since $q(x, X)$ is an $N$-degree polynomial in $x$, directly using the polynomial interpolation for $q(x, X)$ induces at least exponential additive imprecision blowup in $N$ as given in Ref.~\cite{paturi1992degree} (i.e., $\sim e^{O(N)}$), which requires the approximation error $\epsilon$ in $|\text{PGPE}|_{\pm}^{2}$ to be exponentially small ($\sim e^{-O(N)}$) for the reduction.

To avoid this issue, our main strategy is to find a \textit{low-degree polynomial} in $x$, that $(i)$ can be well-estimated via the access to the estimated values of $q(x, X)$ for $x \le x^{*}$ (i.e., access to an oracle that can solve the $|\text{PGPE}|_{\pm}^{2}$ problem), and $(ii)$ contains the desired value $|\Per(X)|^2$ at $x=1$.
Here and throughout, we will use a notation $f_X^{(l)}(x)$ as the low-degree polynomial described above, with $l$ corresponding to the degree of $f_X^{(l)}(x)$.
Then, given the well-estimated values of $f_X^{(l)}(x)$ with each corresponding to $l+1$ number of $x \le x^{*}$ values, we can use the polynomial interpolation for $f_X^{(l)}(x)$, and infer the desired value $f_X^{(l)}(1) = |\Per(X)|^2$. 
This overall process is illustrated in Fig.~\ref{fig:proofoutline}.
Note that the degree $l$ of the low-degree polynomial $f_X^{(l)}(x)$ should be at most $O(\log (\text{poly}(N, \epsilon_0^{-1}, \delta_0^{-1})))$, so that the polynomial interpolation process induces an imprecision blowup that scales at most $e^{O(\log (\text{poly}(N, \epsilon_0^{-1}, \delta_0^{-1})))} = \text{poly}(N, \epsilon_0^{-1}, \delta_0^{-1})$.

\subsection{Conventions}
Before proceeding to the main proof, we establish the parameters that will be frequently used throughout the proof. 
Given the fixed indistinguishability rate $x^{*}$ which is indeed the upper bound of input indistinguishability values $x$, we similarly set the lower bound of $x$ as $x_{\text{min}}$, such that we further restrict the input $x \in [x_{\text{min}},x^{*}]$ to the oracle. 
As the indistinguishability rate $x$ can be expressed as $1 - \frac{k}{N}$ by the average noisy photon number $k$, those boundary values of $x$ can also be represented as
\begin{align}\label{relationbetweenxandk}
    \begin{split}
        x_{\text{min}} = 1 - \frac{k_{\text{max}}}{N}, \quad x^{*} = 1 - \frac{k_{\text{min}}}{N},
    \end{split}
\end{align}
where $k_{\text{min}}$ ($k_{\text{max}}$) denotes the lower (upper) bound of the average noisy photon number $k$, such that $k_{\text{min}} \le k \le k_{\text{max}}$.

For simplicity, as we are interested in logarithmically scaling $k$ with $N$, we use a convention $k = c\log N$ for a constant $c$ without loss of generality. 
Then, the boundary values of $k$ can also be represented as 
\begin{align}\label{relationbetweenkandc}
    k_{\text{max}} = c_{\text{max}}\log N ,\quad k_{\text{min}} = c_{\text{min}}\log N ,
\end{align}
where $c_{\text{min}}$ ($c_{\text{max}}$) denotes the lower (upper) bound of the constant $c$, such that $c_{\text{min}} \le c \le c_{\text{max}}$.

Lastly, since we consider the degree $l$ of the low-degree polynomial $f_X^{(l)}(x)$ scales at most logarithmically with $N$, we similarly define $l$ as
\begin{align}
    l = c_l\log N ,
\end{align}
where we set $l$ larger than any $k$ such that $l > k_{\text{max}}$ (in other words, $c_l > c_{\text{max}}$).
It is worth emphasizing that $c_{\text{min}}$ is a given constant determined by the fixed indistinguishability rate $x^{*}$, whereas $c_{\text{max}}$ and $c_{l}$ are parameters we can arbitrarily set. 


\subsection{Finding a low-degree polynomial $f_X^{(l)}(x)$}

Our key idea is to find a low-degree polynomial $f_X^{(l)}(x)$ that can be well-estimated via the oracle access to $|\text{PGPE}|_{\pm}^{2}$ problem and contains the desired value $|\Per(X)|^2$. 
To obtain such an $l$-degree polynomial $f_X^{(l)}(x)$ from the noisy output probability $q(x, X)$ in Eq.~\eqref{noisyoutputprobability}, our approach is first to truncate the summation terms in $q(x, X)$ and only leave $l+1$ number of terms from $j = N-l$ to $j = N$.

Here, given $l > k_{\text{max}}$, the summation terms in Eq.~\eqref{noisyoutputprobability} from $j=0$ to $j = N-l-1$ correspond to the tail of the binomial distribution $B(N,x)$, whose expectation number of success is $Nx = N - k$. 
Hence, if $l$ is sufficiently larger than $k_{\text{max}}$, one can expect that those tail terms from $j=0$ to $j = N-l-1$ would be small enough.
In other words, we expect that
\begin{align}
    q(x, X)
    &=  \sum_{j=0}^{N}\binom{N}{j}x^{j}(1-x)^{N-j}q_{j}(X) \\
    &\approx  \sum_{j=N-l}^{N}\binom{N}{j}x^{j}(1-x)^{N-j}q_{j}(X),
\end{align}
for $l$ sufficiently larger than $k_{\text{max}}$. 

Hereafter, we newly define $q^{(l)}(x,X)$ as the output probability after truncating the tail terms from $j=0$ to $j = N-l-1$ in $q(x, X)$, such that
\begin{equation}\label{truncatedoutputprobability}
q^{(l)}(x, X) \coloneq \sum_{j=N-l}^{N}\binom{N}{j}x^{j}(1-x)^{N-j}  q_{j}(X) .
\end{equation}
We show that the truncated output probability $q^{(l)}(x,X)$ can be made arbitrarily close to $q(x, X)$, by increasing $l$ sufficiently larger than $k_{\text{max}}$.  
Specifically, we prove the following lemma. 

\begin{lemma}\label{lemma:distancebytruncationissmall}
For $X \sim \mathcal{N}(0,1)_{\mathbb{C}}^{N\times N}$ and $x \ge 1 - \frac{k_{\text{max}}}{N}$, $q(x, X)$ and $q^{(l)}(x,X)$ are $\epsilon_1N!$-close for
\begin{equation}\label{epsilon1}
    \epsilon_1 = \delta^{-1}N^{-\frac{c_{\text{max}}}{3}\left(\frac{c_l}{c_{\text{max}}} - 1 \right)^2 },
\end{equation}
with probability at least $1-\delta$ over $X$.
\end{lemma}
We leave detailed proof of Lemma~\ref{lemma:distancebytruncationissmall} in Appendix~\ref{proofoflemma}. 
One can easily check that $\epsilon_1$ in Eq.~\eqref{epsilon1} decreases with increasing $c_l$. 
Hence, by Lemma~\ref{lemma:distancebytruncationissmall}, $q^{(l)}(x,X)$ can be made close enough to $q(x, X)$ over a large portion of $X \sim \mathcal{N}(0,1)_{\mathbb{C}}^{N\times N}$, by setting $c_l$ sufficiently large compared to $c_{\text{max}}$. 


Given the truncated output probability $q^{(l)}(x,X)$ that is close to the noisy output probability $q(x, X)$, now we can obtain the low-degree polynomial $f_{X}^{(l)}(x)$ by decreasing the polynomial degree of $q^{(l)}(x,X)$.
Specifically, $q^{(l)}(x,X)$ can be expressed as $x^{N-l}f_{X}^{(l)}(x)$ for the $l$-degree polynomial $f_{X}^{(l)}(x)$ such that
\begin{align}
    f_{X}^{(l)}(x) &= x^{-N+l}q^{(l)}(x,X) \\
    &= \sum_{j=N-l}^{N}\binom{N}{j}x^{j-N+l}(1-x)^{N-j}  q_{j}(X)  \\
    &= a_0 + a_1x + a_2x^2 + \dots + a_lx^l ,   
\end{align}
where the coefficients $a_0,\dots,a_l$ of $f_{X}^{(l)}(x)$ are given by the combinations of $q_{j}(X)$ for $j \in [N-l, N]$.

Given that $c_l$ is sufficiently larger than $c_{\text{max}}$, the $l$-degree polynomial $f_{X}^{(l)}(x)$ can be well-estimated for $x \le x^{*}$ via the oracle access to $|\text{PGPE}|_{\pm}^{2}$.
Specifically, this can be done by first querying $q(x,X)$ value (which is close to $q^{(l)}(x,X)$ value as given in Lemma~\ref{lemma:distancebytruncationissmall}) and then multiplying $x^{-N+l}$. 
Also, one can easily check that $f_{X}^{(l)}(x)$ exhibits the desired value $|\Per(X)|^2$ at $x=1$, i.e., $f_{X}^{(l)}(1) = q_{N}(X) = |\Per(X)|^2$ as given in Eq.~\eqref{fixedpdoutputprobability}.
Hence, using the well-estimated values of $f_{X}^{(l)}(x)$ corresponding to $l+1$ number of $x \le x^{*}$ values, we can infer the desired value $f_{X}^{(l)}(1)$ via polynomial interpolation.
This overall process is described in Fig.~\ref{fig:proofoutline}.



\subsection{Establishing reduction from $|\text{GPE}|_{\pm}^{2}$ to $|\text{PGPE}|_{\pm}^{2}$}\label{proofoftheorem1}

Given the $l$-degree polynomial $f_{X}^{(l)}(x)$ as introduced in the previous section, we are now ready to prove Theorem~\ref{logpdhardness}.
Specifically, we establish the BPP-reduction from $|\text{GPE}|_{\pm}^{2}$ to $|\text{PGPE}|_{\pm}^{2}$ for $k_{\rm{min}} = c_{\rm{min}}\log N$, with the error parameters given by $\epsilon =  O\left(\epsilon_0^{7.1}\delta_0^{6.1} N^{-40 c_{\text{min}}}\right)$  and $\delta =  O\left(\frac{\delta_0}{\log N + \log\epsilon_0^{-1}\delta_0^{-1}}\right)$.

\begin{proof}[Proof of Theorem~\ref{logpdhardness}]

Let $\mathcal{O}$ be the oracle introduced in Theorem~\ref{logpdhardness}, i.e., on input $x \le 1-\frac{k_{\text{min}}}{N}$ and $X \sim \mathcal{N}(0,1)_{\mathbb{C}}^{N\times N}$, estimates $q(x, X)$ to within $\epsilon N!$ over $1-\delta$ of $X$ such that 
\begin{equation}
    \Pr_{X}\left[|\mathcal{O}(x,X) - q(x, X)| > \epsilon N! \right] < \delta .
\end{equation}
In the following, we show that given access to the oracle $\mathcal{O}$ one can obtain the desired value $p(1,X) = |\Per(X)|^2$ to within $\epsilon_0 N!$ over $1-\delta_0$ of $X$ (i.e., solve $|\rm{GPE}|_{\pm}^{2}$ problem), for $\epsilon$ and $\delta$ given in Theorem~\ref{logpdhardness}. 

Let us define $\epsilon' \coloneq \epsilon + \epsilon_1$ for $\epsilon_1$ given in Eq.~\eqref{epsilon1}.
By Lemma~\ref{lemma:distancebytruncationissmall}, triangular inequality, and union bound, we have
\begin{align}
    &\Pr_{X}\left[|\mathcal{O}(x,X) - q^{(l)}(x,X)| > \epsilon'N! \right] \nonumber \\
    &\le \Pr_{X}\left[|\mathcal{O}(x,X) - q(x, X)| > \epsilon N! \right] \nonumber \\
    &+ \Pr_{X}\left[|q(x, X) - q^{(l)}(x,X)| > \epsilon_1 N! \right] \\
    &< 2\delta , \label{truncatedprobabilityestimation}
\end{align}
which implies that $\mathcal{O}$ can also well estimate $q^{(l)}(x,X)$ to within $\epsilon'N!$ with probability at least $1 - 2\delta$ over $X$.

As we previously defined, $x_{\text{min}}$ is a minimum value of $x$ we set, such that $x_{\text{min}} = 1 - \frac{k_{\text{max}}}{N}$ with $k_{\text{max}} = c_{\text{max}}\log N$.
Let us define $\epsilon'' \coloneq  \epsilon'N^{c_{\text{max}}}$, such that $\epsilon'x^{-N+l} \le \epsilon'x_{\text{min}}^{-N+l} = \epsilon'e^{k_{\text{max}} - \tilde{\mathcal{O}}(N^{-1})} \le \epsilon''$ for all possible $x$.
Then the following inequality holds:
\begin{align}
    &\Pr_{X}\left[|\mathcal{O}(x,X)x^{-N+l} - f_{X}^{(l)}(x)| > \epsilon''N! \right] \nonumber \\
    &\le \Pr_{X}\left[|\mathcal{O}(x,X)x^{-N+l} - f_{X}^{(l)}(x)| > \epsilon'x^{-N+l}N! \right] \\
    &= \Pr_{X}\left[|\mathcal{O}(x,X) - q^{(l)}(x,X)| > \epsilon'N! \right] \\
    &< 2\delta, \label{estimatelowdegreepolynomial}
\end{align}
by using Eq.~\eqref{truncatedprobabilityestimation}.

Accordingly, for each access to the oracle $\mathcal{O}$ on input a $x \le 1-\frac{k_{\text{min}}}{N}$ and $X \sim \mathcal{N}(0,1)_{\mathbb{C}}^{N\times N}$, one can estimate $f_{X}^{(l)}(x)$ to within $\epsilon''N!$ with probability at least $1 - O(\delta)$ over $X$. 
The remaining problem is to infer the desired value $f_{X}^{(l)}(1)$ from the estimated values $f_{X}^{(l)}(x)$.
Let $\{k_i\}_{i=1}^{l+1}$ be the set of $l+1$ equally spaced points in the interval $k_i\in[k_{\text{min}}, k_{\text{max}}]$.
For each $x_i = 1 - \frac{k_i}{N}$, let $y_i = \mathcal{O}(x_i,X)x_i^{-N+l}$.
By Eq.~\eqref{estimatelowdegreepolynomial}, each set of points $(x_i, y_i)$ satisfies
\begin{equation}
    \Pr\left[|y_i - f_{X}^{(l)}(x_i)| > \epsilon''N! \right] < 2\delta.
\end{equation}
For simplicity, we change the $l$-degree polynomial to $g(x) = f_{X}^{(l)}(z(x))$ for a linear function $z(x) = \frac{k_{\text{max}}+k_{\text{min}}}{2N}(x-1) + 1$, to rescale the input variable $x$.
After rescaling the variable, we have the input $\{x_i\}_{i=1}^{l+1}$ which is the set of equally spaced points in the interval $x_i \in [-\Delta,\Delta]$ with a constant $\Delta \coloneq \frac{k_{\text{max}}-k_{\text{min}}}{k_{\text{max}}+k_{\text{min}}} = \frac{c_{\text{max}}-c_{\text{min}}}{c_{\text{max}}+c_{\text{min}}}$.
Likewise, for each input $x_i$, we have an estimator $y_i =  \mathcal{O}(z(x_i),X)z(x_i)^{-N+l}$ for $g(x_i)$ such that 
\begin{equation}
    \Pr\left[|y_i - g(x_i)| > \epsilon''N! \right] < 2\delta.
\end{equation}
Using those estimation values, we infer the value $g(1) = f_{X}^{(l)}(1) = |\Per(X)|^2$ via the polynomial interpolation technique, as we have described previously.

To infer the value $g(1)$ we use the Lagrange interpolation technique for $g(x)$; to do so, it should be promised that all the $(x_i, y_i)$ points satisfy $|y_i - g(x_i)| \le \epsilon''N!$.
Here, by simply using the union bound, the probability that all the $l+1$ number of $y_i$ points are $\epsilon''N!$-close to $g(x_i)$ is at least $(1-2\delta)^{l+1}$. 
Given that all the $l+1$ points are successful, we can use the Lagrange interpolation technique, whose error's upper bound has been shown in Ref.~\cite{kondo2022quantum} as follows. 

\begin{lemma}[Kondo et al~\cite{kondo2022quantum}]
Let $h(x)$ be a polynomial of degree at most $l$, Let $\Delta\in(0,1)$. Assume that $|h(x_j)|\le \epsilon$ for all of the $l+1$ equally-spaced points $x_j = -\Delta + \frac{2j}{l}\Delta$ for $j = 0,\dots,l$. Then
\begin{equation}\label{interpolationerror}
    |h(1)| < \epsilon\frac{\exp[l(1+\log\Delta^{-1})]}{\sqrt{2\pi l}}.
\end{equation}
\end{lemma}
By drawing $l$-degree polynomial $g_e(x)$ according to $(x_i,y_i)$ points, we get a polynomial $h(x) \coloneq g_e(x) - g(x)$.
Then, $|h(x_i)| \le \epsilon''N!$ for $x_i$ satisfying $|y_i - g(x_i)| \le \epsilon''N!$.
Therefore, given the error bound as in Eq.~\eqref{interpolationerror}, we can obtain an estimator $g_e(1)$ for $g(1)$ satisfying $|g_e(1) - g(1)| < \epsilon''\frac{\exp[l(1+\log\Delta^{-1})]}{\sqrt{2\pi l}}N! < \epsilon''e^{l(1+\log\Delta^{-1})}N!$ via Lagrange interpolation whenever all the $l+1$ points are successful, whose probability is at least $(1-2\delta)^{l+1}$. 
In other words, we can estimate the desired value $g(1) = |\Per(X)|^2$ as
\begin{align}\label{estimationbyinterpolation}
\begin{split}
    \Pr\left[|g_e(1) - |\Per(X)|^2| > \epsilon''e^{l(1+\log\Delta^{-1})}N! \right] \\
    < 1 - (1-2\delta)^{l+1}.
\end{split}
\end{align}
Therefore, to solve $|\text{GPE}|_{\pm}^{2}$ with Eq.~\eqref{estimationbyinterpolation}, the error parameters $\epsilon$ and $\delta$ should satisfy the conditions $\epsilon_0 \ge \epsilon''e^{l(1+\log\Delta^{-1})} =  \epsilon''N^{c_l(1+\log\Delta^{-1})}$ and $\delta_0 \ge 1 - (1-2\delta)^{l+1}$. 
Accordingly, we set $\delta$ satisfying $\delta_0 = 1 - (1-2\delta)^{l+1}$.
Also, combining the results obtained so far, we get
\begin{align}
    \epsilon_0 \ge N^{c_{\text{max}} + c_l(1+\log\Delta^{-1})}\left(\epsilon + \epsilon_1 \right) ,
\end{align}
which requires $\epsilon$ to satisfy the following condition to solve $|\text{GPE}|_{\pm}^{2}$ problem:
\begin{align}\label{relationbtnepsilonandepsilon0}
    \epsilon \le \epsilon_0N^{-c_{\text{max}}-c_l(1+\log\Delta^{-1})} - \epsilon_1,
\end{align}
for $\epsilon_1 = \delta^{-1}N^{-\frac{c_{\text{max}}}{3}\left(\frac{c_l}{c_{\text{max}}} - 1 \right)^2 }$ as given in Lemma~\ref{lemma:distancebytruncationissmall}.

Since the error parameter is positive such that $\epsilon > 0$, one should make sure that $\epsilon_0N^{-c_{\text{max}} -c_l(1+\log\Delta^{-1})}$ is the leading term in the right-hand side of Eq.~\eqref{relationbtnepsilonandepsilon0} as the system size scales. 
To do so, note that we have the freedom to choose the value $c_l$ and $c_{\text{max}}$ for given $c_{\text{min}}$.
Specifically, by setting $c_l$ sufficiently larger than $c_{\text{max}}$, we can make $\epsilon_1$ arbitrarily smaller than $\epsilon_0N^{-c_{\text{max}} -c_l(1+\log\Delta^{-1})}$. 
On the other hand, arbitrarily increasing $c_l$ results in arbitrarily small approximation error $\epsilon$ in Eq.~\eqref{relationbtnepsilonandepsilon0}, because $\epsilon$ scales $\sim N^{-c_l(1+\log\Delta^{-1})}$ with $c_l$.
Therefore, we need to find an \textit{intermediate} size of $c_l$, to (i) make $\epsilon_1$ small enough compared to $\epsilon_0N^{-c_{\text{max}} -c_l(1+\log\Delta^{-1})}$ and (ii) prevent an excessive decrease in $\epsilon$.

Considering all those arguments, we investigate the appropriate size of the parameters $c_{\text{max}}$ and $c_l$ in Appendix~\ref{determiningcmaxandcl}.
More specifically, we set $c_{\text{max}}$ and $c_l$ as 
\begin{align}
    &c_{\text{max}} = \kappa c_{\text{min}} , \label{cmax} \\
    &c_l =  \kappa \lambda \frac{5+3\log\frac{\kappa + 1}{\kappa - 1}}{2} c_{\text{min}} + \frac{3\log\epsilon_0^{-1}\delta_0^{-1}}{\log N} ,
\end{align}
with $\kappa = 2.108$ and $\lambda = 2.149$.
Using these values, we find in Appendix~\ref{determiningcmaxandcl} that the right-hand side of Eq.~\eqref{relationbtnepsilonandepsilon0} becomes
\begin{align}
     &\epsilon_0N^{-c_{\text{max}}-c_l(1+\log\Delta^{-1})} - \epsilon_1 \nonumber \\
     &= O\left( \epsilon_0N^{-c_{\text{max}}-c_l(1+\log\Delta^{-1})} \right)\\
     &= O\left( \epsilon_0^{7.094}\delta_0^{6.094} N^{-39.35 c_{\text{min}}} \right) ,
\end{align}
where we used $\Delta = \frac{c_{\text{max}}-c_{\text{min}}}{c_{\text{max}}+c_{\text{min}}} = \frac{\kappa - 1}{\kappa + 1}$. 
Hence, by setting $c_{\text{max}}$ and $c_l$ as the above, the condition for the error parameter $\epsilon$ in Eq.~\eqref{relationbtnepsilonandepsilon0} can be expressed as 
\begin{align}
    \epsilon \le O\left( \epsilon_0^{7.094}\delta_0^{6.094} N^{-39.35 c_{\text{min}}} \right). 
\end{align}

To sum up, by setting the error parameters $\epsilon$ and $\delta$ as $\delta = 1 - (1-2\delta)^{l+1} = O(\delta_0l^{-1}) = O\left(\frac{\delta_0}{\log N + \log\epsilon_0^{-1}\delta_0^{-1}}\right)$ and 
\begin{align}
\begin{split}
\epsilon &= O\left(\epsilon_0^{7.1}\delta_0^{6.1} N^{-40c_{\text{min}}}\right) ,
\end{split}
\end{align}
one can probabilistically estimate $|\Per(X)|^2$ to within $\epsilon_0N!$ with probability at least $1 - \delta_0$ over $X$ given access to the oracle $\mathcal{O}$.
This completes the proof.

\end{proof}

\section{Classical intractability for partially distinguishable and lossy boson sampling}\label{section: generalization}

Indeed, photon loss and partial distinguishability of photons are typically regarded as the most crucial noise sources in boson sampling experiments~\cite{zhong2020quantum, zhong2021phase, madsen2022quantum, deng2023gaussian}.
Accordingly, in this section we discuss whether the classical intractability of boson sampling can be maintained when $both$ photon loss noise and partial distinguishability noise are applied.

For the photon loss noise, we consider the conventional beamsplitter loss channel described in Refs.~\cite{oszmaniec2018classical, garcia2019simulating, qi2020regimes} where the degree of noise is characterized by a uniform transmission rate $\eta \in [0,1]$ of the entire circuit.
\cor{Let us denote $n$ as the output photon number after photon loss, such that $n \le N$ for the input photon number $N$. }
When the uniform transmission rate $\eta$ is applied for each of $N$ input photons, the probability to observe $n$ photon outcomes follows a binomial distribution $\binom{N}{n}\eta^n(1-\eta)^{N-n}$.

Now, let $q(x,\eta,X)$ be a (rescaled) output probability of obtaining the first $n$ modes outcome over $M$ modes, \cor{where the rescaling factor is given by $M^{n}$.}
Here, $X$ is the first $n \times N$ submatrix of a rescaled $M$ by $M$ unitary matrix $\sqrt{M}U$, which can be approximated to the random Gaussian matrix $X \sim \mathcal{N}(0,1)_{\mathbb{C}}^{n \times N} $ for $M$ by $M$ Haar random unitary matrix $U$ given that $M = \Omega(N^{\gamma})$ for a sufficiently large $\gamma$.
Then, as proposed in Ref.~\cite{moylett2019classically}, $q(x,\eta,X)$ can be represented as
\begin{align}\label{outputprobabilityforlossanddistinguishability}
\begin{split}
    q(x,\eta,X) &= \eta^n(1-\eta)^{N-n} \sum_{\substack{K \subseteq [N]\\|K|=n}} \sum_{j=1}^{n}x^j(1-x)^{N-j} \\
    &\times \sum_{\substack{I \subseteq K \\|I|=j}} \sum_{\substack{J \subseteq K\\|J|=j}} |\Per((X_{K})_{I,J})|^2 \Per(|(X_{K})_{\bar{I},\bar{J}}|^2) \\
    & = \binom{N}{n}\eta^n(1-\eta)^{N-n}  \cdot \binom{N}{n}^{-1}\sum_{\substack{K \subseteq [N]\\|K|=n}} q(x, X_{K}) ,
\end{split}
\end{align}
where $X_{K}$ is an $n$ by $n$ matrix defined by taking columns of $X$ according to $K$.
Also, $q(x, X_{K})$ is the partially distinguishable output probability given in Eq.~\eqref{noisyoutputprobability}.

\cor{

Here, let us remark that we consider the general output probability for an $n$-photon outcome, in the sense that the binomial factor $\binom{N}{n}\eta^n(1-\eta)^{N-n}$ is included in the output probability on the right-hand side of Eq.~\eqref{outputprobabilityforlossanddistinguishability}.
Note that when $n$, $N$, and $\eta$ satisfies $\binom{N}{n}\eta^n(1-\eta)^{N-n} = \Omega( \text{poly}(n)^{-1})$, one can instead consider a noisy boson sampler with a fixed output photon number $n$, which can be implemented by repeating the sampling polynomial number of times and post-selecting $n$-photon outcomes.
In this case, the binomial factor $\binom{N}{n}\eta^n(1-\eta)^{N-n}$ in the right-hand side of Eq.~\eqref{outputprobabilityforlossanddistinguishability} vanishes, yielding a conditional output probability that is independent of $\eta$.
However, since our analysis of the condition on $\eta$ for the hardness result requires an explicit dependence of the output probability on $\eta$, we use the general output probability in Eq.~\eqref{outputprobabilityforlossanddistinguishability} rather than the conditional output probability.

}

To examine the classical intractability of simulating noisy boson sampling with loss and partial distinguishability, we follow the approach in Section~\ref{section:problemset}.
That is, we similarly set the problem corresponding to the average-case estimation of the noisy output probability $q(x,\eta,X)$ over $X \sim \mathcal{N}(0,1)_{\mathbb{C}}^{n \times N}$, which can be formally stated as follows.

\begin{problem}[$\sum|\text{PGPE}|_{\pm}^{2}$]
Given as input a transmission rate $\eta \in [0, \eta^{*}]$, an indistinguishability rate $x \in [0, x^{*}]$, and a matrix $X \sim \mathcal{N}(0,1)_{\mathbb{C}}^{n \times N}$ of i.i.d. Gaussians with $N \ge n$, together with error bounds $\epsilon,\,\delta > 0$, estimate $q(x,\eta,X)$ to within additive error $\pm \epsilon\cdot n!$ with probability at least $1-\delta$ over $X$ in $\text{poly}(n,\epsilon^{-1},\delta^{-1})$ time. 
\end{problem}

\cor{Similarly to $x^{*}$, $\eta^{*}$ is a fixed transmission rate of the noisy boson sampling, such that $1 - x^{*}$ and $1 - \eta^{*}$ characterize the noise rate for partial distinguishability and photon loss noises, respectively.}
Also, to avoid any confusion, we note that now the problem size is given by the output photon number $n$, 
while we set the input photon number $N$ as a free parameter that we can arbitrarily choose, as long as $n \leq N \ll M$.

\cor{

As we have previously argued in Section~\ref{section:problemset}, \#P-hardness of $\sum|\text{PGPE}|_{\pm}^{2}$ problem serves as a sufficient condition for the classical intractability of simulating the noisy boson sampling subject to loss and distinguishability. 
Specifically, on input a Haar-random unitary matrix $U$, $\eta \in [0, \eta^{*}]$ and $x \in [0, x^{*}]$, the output probability of the noisy boson sampling for the first $n$ modes outcome can be represented as $q(x,\eta, X)$ in Eq.~\eqref{outputprobabilityforlossanddistinguishability} (rescaled by $M^{n}$), for $X \sim \mathcal{N}(0,1)_{\mathbb{C}}^{n \times N}$.
Then, considering a randomized classical sampler $\mathcal{S}$ that simulates the noisy boson sampling to within total variation distance $\beta$, with $\eta^{*}$ and $x^{*}$ characterizing the noise rate for loss and partial distinguishability, $\sum|\text{PGPE}|_{\pm}^{2}$ is in complexity class $\rm{BPP}^{NP^{\mathcal{S}}}$ for $\beta = O(\epsilon\delta)$ by Stockmeyer's reduction~\cite{stockmeyer1985approximation}.
Therefore, showing the \#P-hardness of $\sum|\text{PGPE}|_{\pm}^{2}$ implies that there cannot be any efficient classical sampler $\mathcal{S}$ for the approximate noisy boson sampling, unless the polynomial hierarchy collapses.
We note that this hardness result requires that $n$, $N$, and $\eta$ satisfy $\binom{N}{n}\eta^n(1-\eta)^{N-n} = \Omega( \text{poly}(n)^{-1})$, ensuring that $n$ photon outcomes that we are interested in occupy a sufficiently large portion of probability weight over all possible outcomes. 

}

Following our main strategy, to show the \#P-hardness of $\sum|\text{PGPE}|_{\pm}^{2}$, one way is to find whether $\sum|\text{PGPE}|_{\pm}^{2}$ can be reduced from $|\text{GPE}|_{\pm}^{2}$, showing that solving $\sum|\text{PGPE}|_{\pm}^{2}$ is at least \textit{as hard as} solving $|\text{GPE}|_{\pm}^{2}$.
More precisely, we investigate whether average-case estimation of $q(n,x,X)$ is polynomially reducible from the average-case estimation of ideal output probability $q(X_{K_0}) = |\Per(X_{K_0})|^2$, for some fixed subset of modes $K_0 \subseteq [N]$ with $|K_0| = n$, corresponding to the $|\text{GPE}|_{\pm}^{2}$ problem with input $X_{K_0} \sim \mathcal{N}(0,1)_{\mathbb{C}}^{n\times n}$.
If this reduction can be established, it implies that such a noisy boson sampling is equivalently hard to classically simulate as the ideal boson sampling. 
\cor{Since we have already established the reduction from $|\text{GPE}|_{\pm}^{2}$ to $|\text{PGPE}|_{\pm}^{2}$ for $O(\log n)$ number of distinguishable photons in Sec.~\ref{Section:reduction}, the remaining problem is to establish the reduction from $|\text{PGPE}|_{\pm}^{2}$ to $\sum|\text{PGPE}|_{\pm}^{2}$, i.e., estimate $q(x,X_{K_0})$ in Eq.~\eqref{noisyoutputprobability} for a fixed $K_0$ from estimated values of $q(x,\eta,X)$ in Eq.~\eqref{outputprobabilityforlossanddistinguishability}.

To this end, a natural approach is to set $n$ close to the mean output photon number (such that $N \simeq n/\eta$), which always satisfies the requirement $\binom{N}{n}\eta^n(1-\eta)^{N-n} = \Omega(\text{poly}(n)^{-1})$ for classical hardness.
In this case, one can employ the previous result in Ref.~\cite{aaronson2016bosonsampling} for lossy boson sampling.
Specifically, for $X \sim \mathcal{N}(0,1)_{\mathbb{C}}^{n \times N}$, Ref.~\cite{aaronson2016bosonsampling} has established a reduction from the average-case estimation of $q(X_{K_0})= |\Per(X_{K_0})|^2$ to the average-case estimation of $\binom{N}{n}^{-1}\sum_{K} q(X_{K})$ (where $K \subseteq [N]$ with $|K| = n$) for $N = n + k'$, with imprecision blowup by at most $\sim n^{O(k')}$
Hence, it is polynomially reducible when $k' = O(1)$, denoting that a constant number of lost photons maintains the computational complexity of lossy boson sampling.
By extending their techniques to our case and setting $N \simeq n/\eta$, one can similarly establish a reduction from the average-case extimation of $q(x,X_{K_0})$ to the average-case estimation of $q(x,\eta,X)$, with imprecision blowup by at most $\sim n^{O(k')}$, where now $k' \simeq (1-\eta)n/\eta$ denotes the average number of lost photons for transmission rate $\eta$.
Here, the imprecision blowup is minimized when $\eta = \eta^{*}$, because $(1-\eta)n/\eta \geq (1 - \eta^{*})n/\eta^{*}$. 
Therefore, given that $(1 - \eta^{*})n/\eta^{*} = O(1)$, which implies a constant number of photon loss on average, $\sum|\text{PGPE}|_{\pm}^{2}$ is polynomially reducible from $|\text{PGPE}|_{\pm}^{2}$. 
And combining our result, $\sum|\text{PGPE}|_{\pm}^{2}$ is polynomially reducible from $|\text{GPE}|_{\pm}^{2}$, for $O(1)$ number of loss and $O(\log n)$ number of distinguishable photons on average. 
}

\cor{
We remark that there's a much simpler way, as also noted in Ref.~\cite{aaronson2016bosonsampling}.
That is, one can instead set $n = N$, corresponding to the \textit{lossless} output probability.
More specifically, from Eq.~\eqref{outputprobabilityforlossanddistinguishability}, the output probability for $n = N$ is given by
\begin{align}
     q(x,\eta,X) = \eta^{N}  q(x, X),
\end{align}
for $X \sim \mathcal{N}(0,1)_{\mathbb{C}}^{N \times N}$, and $q(x,X)$ is the partially distinguishable output probability in Eq.~\eqref{noisyoutputprobability}.
Hence, one can directly estimate $q(x, X)$ by estimating $q(x,\eta,X)$, but with a multiplicative imprecision blowup by $\eta^{-N}$.
Observe that, $\eta^{-N}$ is minimized when $\eta = \eta^{*}$, and when $\eta^{*}$ satisfies $(1-\eta^{*})n = O(\log n)$, we have $(\eta^{*})^{-N} = e^{O(\log n)} = \text{poly}(n)$.
Indeed, this indicates $(1-\eta^{*})n/\eta^{*} =   O(\log n / \eta^{*}) =  O(\log n)$ number of photon loss on average.
Therefore, under $O(\log n)$ number of loss, $\sum|\text{PGPE}|_{\pm}^{2}$ is polynomially reducible from $|\text{PGPE}|_{\pm}^{2}$, with the required imprecision level decreases multiplicatively by $\text{poly}(n)$ during the reduction.
Combining this argument with Theorem~\ref{logpdhardness}, for on average $O(\log n)$ number of loss and $O(\log n)$ number of distinguishability, $\sum|\text{PGPE}|_{\pm}^{2}$ is polynomially reducible from $|\text{GPE}|_{\pm}^{2}$, given that $\epsilon, \delta = O({\rm{poly}}(n, \epsilon_0^{-1}, \delta_0^{-1})^{-1})$.
Although this simple approach can yield a stronger hardness result, it cannot be extended to higher loss rates, since $(\eta^{*})^{-N}$ becomes quasi-polynomial in $n$ outside the log-loss regime.
Therefore, this exhibits clear limitations beyond the log-loss regime, and establishing hardness results under higher loss rates necessitates fundamentally different techniques.

We now conclude this section by formally summarizing the arguments so far.

}

\begin{corollary}
If $\mathcal{O}$ is an oracle that solves $\sum|\rm{PGPE}|_{\pm}^{2}$ for $(1-\eta^{*})n = O(\log n)$ and $(1-x^{*})n = O(\log n)$, together with $\epsilon, \delta = O({\rm{poly}}(n, \epsilon_0^{-1}, \delta_0^{-1})^{-1})$, then $|\rm{GPE}|_{\pm}^{2}$ can be solved in $\rm{BPP}^{\mathcal{O}}$. 
\end{corollary}


\section{Concluding remarks}\label{section: conclusion}

In this work, we showed that $O(\log N)$ number of distinguishable photons on average for $N$ photon input maintains an equivalent complexity to the ideal boson sampling case.
Specifically, we established a complexity-theoretical reduction from the problem of estimating the output probabilities of ideal boson sampling, to the problem of estimating the output probabilities of noisy boson sampling with a logarithmic number of distinguishable photons on average. 
This implies that noisy boson sampling with $O(\log N)$ number of distinguishable photons preserves the classical intractability of the ideal boson sampling. 
We also extended our result to the partially distinguishable and lossy boson sampling and found that $O(\log N)$ number of photon loss and $O(\log N)$ number of distinguishable photons upholds the classical intractability of the ideal boson sampling.

Let us first remark on the barrier to using our main technique in the photon loss case, i.e., noisy boson sampling subject to photon loss.   
Our main strategy is that, given access to the noisy output probability for different noise rates, we first conduct polynomial interpolation to infer the noisy output probability in terms of the noise rate, and then obtain the ideal output probability by setting the noise rate to zero. 
Therefore, the requirement of the noisy output probability to use this approach is that it becomes the ideal output probability when the noise rate goes to zero. 
For the partial distinguishability noise, the noisy output probability $q(x, X)$ in Eq.~\eqref{noisyoutputprobability} satisfies this condition, because it exhibits the ideal output probability at $x = 1$.
However, for the photon loss case with transmission rate $\eta$, the noisy output probability to obtain the first $n$ modes outcome for $N \ge n$ input photons (without post-selection) is~\cite{oszmaniec2018classical} 
\begin{equation}
q(\eta, X) = \eta^n(1-\eta)^{N-n}\sum_{\substack{K \subseteq [N]\\|K|=n}} |\Per(X_{K})|^2 ,
\end{equation}
which does not converge to the ideal output probability by only adjusting the transmission rate $\eta$, except for the output probability without any loss (i.e., $N = n$, which is the ideal output probability itself). 
Moreover, when we post-select $n$ photon outcomes as in Ref.~\cite{aaronson2016bosonsampling}, the noisy output probability loses its dependence on $\eta$, and thus we cannot take our approach for this case either. 
Therefore, our main strategy does not work for the output probability of lossy boson sampling.

We also remark that in contrast to the photon loss case, we cannot employ the `post-selection trick' for the partial distinguishability noise, i.e., post-selecting noise-free event. 
For the photon loss case, we can post-select no-loss outcomes by repeating the sampling process until the output photon number of the outcome is equal to the input photon number. 
As described in Sec.~\ref{section: generalization}, the probability of post-selecting no-loss outcome for $k'$ number of lost photons on average for input $N$ photons is $(1-\frac{k'}{N})^N$. 
Hence, via post-selection, one can directly deduce that lossy boson sampling with at most $k' = O(\log N)$ number of lost photons on average remains its complexity to the original boson sampling case. 
On the other hand, in the case of partial distinguishability noise, we cannot post-select the noise-free outcome, i.e., the outcome without any distinguishable photons. 
Instead, the output probability of partially distinguishable boson sampling is represented as a summation of output probabilities corresponding to a fixed number of distinguishable photons from $0$ to $N$, as depicted in Eq.~\eqref{noisyoutputprobability}.
Therefore, in contrast to the photon loss case, we cannot directly obtain the classical hardness result for $O(\log N)$ number of noisy photons via post-selection in the partial distinguishability case.

We now conclude with a few open problems.

1. A crucial open problem that needs to be addressed is the extension of our result for a larger rate of noise.
It is reasonable to expect that the experimental implementation of boson sampling generates at least $\Theta(N)$ noisy photons over $N$ input photons.
That is, a constant fraction of input photons suffer from noise, and for our case, this corresponds to $x = \Theta(1)$.
In this regime, the low-degree approximation of the noisy output probability in Eq.~\eqref{noisyoutputprobability}, which was our main strategy for the proof of Theorem~\ref{logpdhardness} in Sec.~\ref{Section:reduction}, does not hold anymore. 
However, as long as we use the polynomial interpolation to infer the noiseless output probability, it is crucial to decrease the degree of the polynomial corresponding to the noisy output probability to at most $O(\log (\text{poly}(N, \epsilon_0^{-1}, \delta_0^{-1})))$ degree as we have described previously. 
Hence, a more advanced technique beyond our low-degree approximation has to be developed to extend our result to larger noise rates; we leave this problem as an open question.

2. Another open problem is investigating the classical simulation hardness of partially distinguishable Gaussian boson sampling. 
Gaussian boson sampling is a variant of boson sampling with Gaussian state inputs instead of Fock states, which is more suitable for experimental setups while having similar computational complexity~\cite{hamilton2017gaussian, kruse2019detailed, deshpande2022quantum, grier2022complexity}.
Accordingly, recent experiments have been held on Gaussian boson sampling for the experimental demonstration of quantum computational advantage~\cite{zhong2020quantum, zhong2021phase, madsen2022quantum, deng2023gaussian}.
Therefore, to achieve quantum computational advantage via near-term Gaussian boson sampling experiments, it is crucial to verify that noisy Gaussian boson sampling is classically hard. 
One way to do this is to find that the noisy Gaussian boson sampling is as hard as the ideal Gaussian boson sampling case, similar to Ref.~\cite{aaronson2016bosonsampling} and our main result in this work. 
For the partial distinguishability noise, Ref.~\cite{shi2022effect} identified the noisy output probability of Gaussian boson sampling when partial distinguishability noise is applied. 
However, since our proof technique cannot be directly applied to the noisy output probability in~\cite{shi2022effect}, other proof techniques should be developed to deal with it.

3. One can also consider more realistic noise models, i.e., beyond the uniform distinguishability noise parameterized by a single parameter $x$ as in~\cite{tichy2015sampling}.  
One example is non-uniform distinguishability noise parameterized by $x_1,\dots,x_N$ for each of $N$ input photons. 
In this case, since the output probability cannot be represented as a polynomial of a single parameter (as $x$ in Eq.~\eqref{noisyoutputprobability}), our technique cannot be directly applied. 
One suboptimal approach we can take is to arbitrarily increase each $x_i$ value for all $i \in [N]$, to the maximum value among $x_1,\dots,x_N$.
Then, by defining the maximum value as a parameter $x$, we can use the same process as in the main result and obtain the same result.
Another approach is to employ the results in Refs.~\cite{brod2020classical, van2024efficient}, which dealt with the classical simulability of noisy boson sampling subject to non-uniform physical noises including loss and distinguishability. 
Although those results are for classical simulation algorithms, extending their techniques to the hardness arguments of noisy boson sampling with non-uniform types of noise would be an interesting open problem. 


\acknowledgements

B.G. was supported by the education and training program of the
Quantum Information Research Support Center, funded through the National
Research Foundation of Korea (NRF) by the Ministry of Science and ICT (MSIT) of
the Korean government (No.2021M3H3A1036573).
B.G. and H.J. were supported by the Korean government (Ministry of Science and ICT~(MSIT)),
the NRF grants funded by the Korea government~(MSIT)~(Nos.~RS-2024-00413957,~RS-2024-00438415, and NRF-2023R1A2C1006115),~and the Institute of Information \& Communications Technology Planning \& Evaluation (IITP) grant funded by the Korea government (MSIT) (IITP-2025-RS-2020-II201606 and  IITP-2025-RS-2024-00437191).
C.O. was supported by the NRF Grants (No. RS-2024-00431768 and No. RS-2025-00515456) funded by the Korean government (MSIT) and IITP grants funded by the Korea government (MSIT) (No. IITP-2025-RS-2025-02283189 and  IITP-2025-RS-2025-02263264).

\bibliographystyle{unsrt}
\bibliography{reference}

\begin{thebibliography}{10}

\bibitem{preskill2018quantum}
John Preskill.
\newblock Quantum computing in the nisq era and beyond.
\newblock {\em Quantum}, 2:79, 2018.

\bibitem{feynman2018simulating}
Richard~P Feynman.
\newblock Simulating physics with computers.
\newblock In {\em Feynman and computation}, pages 133--153. cRc Press, 2018.

\bibitem{lloyd1996universal}
Seth Lloyd.
\newblock Universal quantum simulators.
\newblock {\em Science}, 273(5278):1073--1078, 1996.

\bibitem{shor1999polynomial}
Peter~W Shor.
\newblock Polynomial-time algorithms for prime factorization and discrete
  logarithms on a quantum computer.
\newblock {\em SIAM review}, 41(2):303--332, 1999.

\bibitem{bouland2019complexity}
Adam Bouland, Bill Fefferman, Chinmay Nirkhe, and Umesh Vazirani.
\newblock On the complexity and verification of quantum random circuit
  sampling.
\newblock {\em Nature Physics}, 15(2):159--163, 2019.

\bibitem{aaronson2011computational}
Scott Aaronson and Alex Arkhipov.
\newblock The computational complexity of linear optics.
\newblock In {\em Proceedings of the forty-third annual ACM symposium on Theory
  of computing}, pages 333--342, 2011.

\bibitem{hamilton2017gaussian}
Craig~S Hamilton, Regina Kruse, Linda Sansoni, Sonja Barkhofen, Christine
  Silberhorn, and Igor Jex.
\newblock Gaussian boson sampling.
\newblock {\em Physical review letters}, 119(17):170501, 2017.

\bibitem{deshpande2022quantum}
Abhinav Deshpande, Arthur Mehta, Trevor Vincent, Nicol{\'a}s Quesada, Marcel
  Hinsche, Marios Ioannou, Lars Madsen, Jonathan Lavoie, Haoyu Qi, Jens Eisert,
  et~al.
\newblock Quantum computational advantage via high-dimensional gaussian boson
  sampling.
\newblock {\em Science advances}, 8(1):eabi7894, 2022.

\bibitem{grier2022complexity}
Daniel Grier, Daniel~J Brod, Juan~Miguel Arrazola, Marcos~Benicio
  de~Andrade~Alonso, and Nicol{\'a}s Quesada.
\newblock The complexity of bipartite gaussian boson sampling.
\newblock {\em Quantum}, 6:863, 2022.

\bibitem{zhong2020quantum}
Han-Sen Zhong, Hui Wang, Yu-Hao Deng, Ming-Cheng Chen, Li-Chao Peng, Yi-Han
  Luo, Jian Qin, Dian Wu, Xing Ding, Yi~Hu, et~al.
\newblock Quantum computational advantage using photons.
\newblock {\em Science}, 370(6523):1460--1463, 2020.

\bibitem{zhong2021phase}
Han-Sen Zhong, Yu-Hao Deng, Jian Qin, Hui Wang, Ming-Cheng Chen, Li-Chao Peng,
  Yi-Han Luo, Dian Wu, Si-Qiu Gong, Hao Su, et~al.
\newblock Phase-programmable {G}aussian boson sampling using stimulated
  squeezed light.
\newblock {\em Physical review letters}, 127(18):180502, 2021.

\bibitem{madsen2022quantum}
Lars~S Madsen, Fabian Laudenbach, Mohsen~Falamarzi Askarani, Fabien Rortais,
  Trevor Vincent, Jacob~FF Bulmer, Filippo~M Miatto, Leonhard Neuhaus, Lukas~G
  Helt, Matthew~J Collins, et~al.
\newblock Quantum computational advantage with a programmable photonic
  processor.
\newblock {\em Nature}, 606(7912):75--81, 2022.

\bibitem{deng2023gaussian}
Yu-Hao Deng, Yi-Chao Gu, Hua-Liang Liu, Si-Qiu Gong, Hao Su, Zhi-Jiong Zhang,
  Hao-Yang Tang, Meng-Hao Jia, Jia-Min Xu, Ming-Cheng Chen, et~al.
\newblock Gaussian boson sampling with pseudo-photon-number-resolving detectors
  and quantum computational advantage.
\newblock {\em Physical review letters}, 131(15):150601, 2023.

\bibitem{kalai2014gaussian}
Gil Kalai and Guy Kindler.
\newblock Gaussian noise sensitivity and bosonsampling.
\newblock {\em arXiv preprint arXiv:1409.3093}, 2014.

\bibitem{renema2018classical}
Jelmer Renema, Valery Shchesnovich, and Raul Garcia-Patron.
\newblock Classical simulability of noisy boson sampling.
\newblock {\em arXiv preprint arXiv:1809.01953}, 2018.

\bibitem{renema2018efficient}
Jelmer~J Renema, Adrian Menssen, William~R Clements, Gil Triginer, William~S
  Kolthammer, and Ian~A Walmsley.
\newblock Efficient classical algorithm for boson sampling with partially
  distinguishable photons.
\newblock {\em Physical review letters}, 120(22):220502, 2018.

\bibitem{moylett2019classically}
Alexandra~E Moylett, Ra{\'u}l Garc{\'\i}a-Patr{\'o}n, Jelmer~J Renema, and
  Peter~S Turner.
\newblock Classically simulating near-term partially-distinguishable and lossy
  boson sampling.
\newblock {\em Quantum Science and Technology}, 5(1):015001, 2019.

\bibitem{shchesnovich2019noise}
Valery~S Shchesnovich.
\newblock Noise in boson sampling and the threshold of efficient classical
  simulatability.
\newblock {\em Physical Review A}, 100(1):012340, 2019.

\bibitem{renema2020simulability}
Jelmer~J Renema.
\newblock Simulability of partially distinguishable superposition and gaussian
  boson sampling.
\newblock {\em Physical Review A}, 101(6):063840, 2020.

\bibitem{shi2022effect}
Junheng Shi and Tim Byrnes.
\newblock Effect of partial distinguishability on quantum supremacy in gaussian
  boson sampling.
\newblock {\em npj Quantum Information}, 8(1):54, 2022.

\bibitem{van2024efficient}
SN~van~den Hoven, Edwin Kanis, and Jelmer~Jan Renema.
\newblock Efficient classical algorithm for simulating boson sampling with
  inhomogeneous partial distinguishability.
\newblock {\em arXiv preprint arXiv:2406.17682}, 2024.

\bibitem{oszmaniec2018classical}
Micha{\l} Oszmaniec and Daniel~J Brod.
\newblock Classical simulation of photonic linear optics with lost particles.
\newblock {\em New Journal of Physics}, 20(9):092002, 2018.

\bibitem{garcia2019simulating}
Ra{\'u}l Garc{\'\i}a-Patr{\'o}n, Jelmer~J Renema, and Valery Shchesnovich.
\newblock Simulating boson sampling in lossy architectures.
\newblock {\em Quantum}, 3:169, 2019.

\bibitem{qi2020regimes}
Haoyu Qi, Daniel~J Brod, Nicol{\'a}s Quesada, and Ra{\'u}l
  Garc{\'\i}a-Patr{\'o}n.
\newblock Regimes of classical simulability for noisy gaussian boson sampling.
\newblock {\em Physical review letters}, 124(10):100502, 2020.

\bibitem{brod2020classical}
Daniel~Jost Brod and Micha{\l} Oszmaniec.
\newblock Classical simulation of linear optics subject to nonuniform losses.
\newblock {\em Quantum}, 4:267, 2020.

\bibitem{villalonga2021efficient}
Benjamin Villalonga, Murphy~Yuezhen Niu, Li~Li, Hartmut Neven, John~C Platt,
  Vadim~N Smelyanskiy, and Sergio Boixo.
\newblock Efficient approximation of experimental gaussian boson sampling.
\newblock {\em arXiv preprint arXiv:2109.11525}, 2021.

\bibitem{bulmer2022boundary}
Jacob~FF Bulmer, Bryn~A Bell, Rachel~S Chadwick, Alex~E Jones, Diana Moise,
  Alessandro Rigazzi, Jan Thorbecke, Utz-Uwe Haus, Thomas Van~Vaerenbergh,
  Raj~B Patel, et~al.
\newblock The boundary for quantum advantage in gaussian boson sampling.
\newblock {\em Science advances}, 8(4):eabl9236, 2022.

\bibitem{oh2024classicalalgorithm}
Changhun Oh, Minzhao Liu, Yuri Alexeev, Bill Fefferman, and Liang Jiang.
\newblock Classical algorithm for simulating experimental gaussian boson
  sampling.
\newblock {\em Nature Physics}, pages 1--8, 2024.

\bibitem{oh2023classical}
Changhun Oh, Liang Jiang, and Bill Fefferman.
\newblock On classical simulation algorithms for noisy boson sampling.
\newblock {\em arXiv preprint arXiv:2301.11532}, 2023.

\bibitem{oh2024classical}
Changhun Oh.
\newblock Classical simulability of constant-depth linear-optical circuits with
  noise.
\newblock {\em arXiv preprint arXiv:2406.08086}, 2024.

\bibitem{oh2025recent}
Changhun Oh.
\newblock Recent theoretical and experimental progress on boson sampling.
\newblock {\em Current Optics and Photonics}, 9(1):1--18, 2025.

\bibitem{bouland2022noise}
Adam Bouland, Bill Fefferman, Zeph Landau, and Yunchao Liu.
\newblock Noise and the frontier of quantum supremacy.
\newblock In {\em 2021 IEEE 62nd Annual Symposium on Foundations of Computer
  Science (FOCS)}, pages 1308--1317. IEEE, 2022.

\bibitem{bouland2023complexity}
Adam Bouland, Daniel Brod, Ishaun Datta, Bill Fefferman, Daniel Grier, Felipe
  Hernandez, and Michal Oszmaniec.
\newblock Complexity-theoretic foundations of bosonsampling with a linear
  number of modes.
\newblock {\em arXiv preprint arXiv:2312.00286}, 2023.

\bibitem{bouland2024average}
Adam Bouland, Ishaun Datta, Bill Fefferman, and Felipe Hernandez.
\newblock On the average-case hardness of bosonsampling.
\newblock {\em arXiv preprint arXiv:2411.04566}, 2024.

\bibitem{aaronson2016bosonsampling}
Scott Aaronson and Daniel~J Brod.
\newblock Bosonsampling with lost photons.
\newblock {\em Physical Review A}, 93(1):012335, 2016.

\bibitem{chen2023scalable}
Wentao Chen, Yao Lu, Shuaining Zhang, Kuan Zhang, Guanhao Huang, Mu~Qiao,
  Xiaolu Su, Jialiang Zhang, Jing-Ning Zhang, Leonardo Banchi, et~al.
\newblock Scalable and programmable phononic network with trapped ions.
\newblock {\em Nature Physics}, pages 1--7, 2023.

\bibitem{young2024atomic}
Aaron~W Young, Shawn Geller, William~J Eckner, Nathan Schine, Scott Glancy,
  Emanuel Knill, and Adam~M Kaufman.
\newblock An atomic boson sampler.
\newblock {\em Nature}, 629(8011):311--316, 2024.

\bibitem{shchesnovich2014sufficient}
VS~Shchesnovich.
\newblock Sufficient condition for the mode mismatch of single photons for
  scalability of the boson-sampling computer.
\newblock {\em Physical Review A}, 89(2):022333, 2014.

\bibitem{tichy2015sampling}
Malte~C Tichy.
\newblock Sampling of partially distinguishable bosons and the relation to the
  multidimensional permanent.
\newblock {\em Physical Review A}, 91(2):022316, 2015.

\bibitem{hong1987measurement}
Chong-Ki Hong, Zhe-Yu Ou, and Leonard Mandel.
\newblock Measurement of subpicosecond time intervals between two photons by
  interference.
\newblock {\em Physical review letters}, 59(18):2044, 1987.

\bibitem{jerrum2004polynomial}
Mark Jerrum, Alistair Sinclair, and Eric Vigoda.
\newblock A polynomial-time approximation algorithm for the permanent of a
  matrix with nonnegative entries.
\newblock {\em Journal of the ACM (JACM)}, 51(4):671--697, 2004.

\bibitem{stockmeyer1985approximation}
Larry Stockmeyer.
\newblock On approximation algorithms for\# p.
\newblock {\em SIAM Journal on Computing}, 14(4):849--861, 1985.

\bibitem{cai2023quantum}
Zhenyu Cai, Ryan Babbush, Simon~C Benjamin, Suguru Endo, William~J Huggins,
  Ying Li, Jarrod~R McClean, and Thomas~E O’Brien.
\newblock Quantum error mitigation.
\newblock {\em Reviews of Modern Physics}, 95(4):045005, 2023.

\bibitem{temme2017error}
Kristan Temme, Sergey Bravyi, and Jay~M Gambetta.
\newblock Error mitigation for short-depth quantum circuits.
\newblock {\em Physical review letters}, 119(18):180509, 2017.

\bibitem{li2017efficient}
Ying Li and Simon~C Benjamin.
\newblock Efficient variational quantum simulator incorporating active error
  minimization.
\newblock {\em Physical Review X}, 7(2):021050, 2017.

\bibitem{paturi1992degree}
Ramamohan Paturi.
\newblock On the degree of polynomials that approximate symmetric boolean
  functions (preliminary version).
\newblock In {\em Proceedings of the twenty-fourth annual ACM symposium on
  Theory of computing}, pages 468--474, 1992.

\bibitem{kondo2022quantum}
Yasuhiro Kondo, Ryuhei Mori, and Ramis Movassagh.
\newblock Quantum supremacy and hardness of estimating output probabilities of
  quantum circuits.
\newblock In {\em 2021 IEEE 62nd Annual Symposium on Foundations of Computer
  Science (FOCS)}, pages 1296--1307. IEEE, 2022.

\bibitem{kruse2019detailed}
Regina Kruse, Craig~S Hamilton, Linda Sansoni, Sonja Barkhofen, Christine
  Silberhorn, and Igor Jex.
\newblock Detailed study of gaussian boson sampling.
\newblock {\em Physical Review A}, 100(3):032326, 2019.

\bibitem{rahimi2016sufficient}
Saleh Rahimi-Keshari, Timothy~C Ralph, and Carlton~M Caves.
\newblock Sufficient conditions for efficient classical simulation of quantum
  optics.
\newblock {\em Physical Review X}, 6(2):021039, 2016.

\bibitem{vidal2003efficient}
Guifr{\'e} Vidal.
\newblock Efficient classical simulation of slightly entangled quantum
  computations.
\newblock {\em Physical review letters}, 91(14):147902, 2003.

\bibitem{qi2022efficient}
Haoyu Qi, Diego Cifuentes, Kamil Br{\'a}dler, Robert Israel, Timjan
  Kalajdzievski, and Nicol{\'a}s Quesada.
\newblock Efficient sampling from shallow gaussian quantum-optical circuits
  with local interactions.
\newblock {\em Physical Review A}, 105(5):052412, 2022.

\bibitem{deshpande2018dynamical}
Abhinav Deshpande, Bill Fefferman, Minh~C Tran, Michael Foss-Feig, and Alexey~V
  Gorshkov.
\newblock Dynamical phase transitions in sampling complexity.
\newblock {\em Physical review letters}, 121(3):030501, 2018.

\bibitem{oh2022classical}
Changhun Oh, Youngrong Lim, Bill Fefferman, and Liang Jiang.
\newblock Classical simulation of boson sampling based on graph structure.
\newblock {\em Physical Review Letters}, 128(19):190501, 2022.

\bibitem{kolarovszki2023simulating}
Zolt{\'a}n Kolarovszki, {\'A}goston Kaposi, Tam{\'a}s Kozsik, and Zolt{\'a}n
  Zimbor{\'a}s.
\newblock Simulating sparse and shallow gaussian boson sampling.
\newblock In {\em International Conference on Computational Science}, pages
  209--223. Springer, 2023.

\bibitem{go2024exploring}
Byeongseon Go, Changhun Oh, Liang Jiang, and Hyunseok Jeong.
\newblock Exploring shallow-depth boson sampling: Toward a scalable quantum
  advantage.
\newblock {\em Physical Review A}, 109(5):052613, 2024.

\bibitem{go2024computational}
Byeongseon Go, Changhun Oh, and Hyunseok Jeong.
\newblock On computational complexity and average-case hardness of
  shallow-depth boson sampling.
\newblock {\em arXiv preprint arXiv:2405.01786}, 2024.

\bibitem{clements2016optimal}
William~R Clements, Peter~C Humphreys, Benjamin~J Metcalf, W~Steven Kolthammer,
  and Ian~A Walmsley.
\newblock Optimal design for universal multiport interferometers.
\newblock {\em Optica}, 3(12):1460--1465, 2016.

\bibitem{russell2017direct}
Nicholas~J Russell, Levon Chakhmakhchyan, Jeremy~L O’Brien, and Anthony
  Laing.
\newblock Direct dialling of haar random unitary matrices.
\newblock {\em New journal of physics}, 19(3):033007, 2017.

\bibitem{dubhashi2009concentration}
Devdatt~P Dubhashi and Alessandro Panconesi.
\newblock {\em Concentration of measure for the analysis of randomized
  algorithms}.
\newblock Cambridge University Press, 2009.

\end{thebibliography}

\appendix

\cor{
\section{Comparison to previous results}\label{appendix: section: related works}

In this appendix, we provide a detailed overview of the existing literature on the complexity of noisy boson sampling and compare the presented work to these findings to further clarify our contribution.

\subsection{Our main result}

We first summarize our main results, emphasizing our key assumption, main technique, and noise model. 
In the present work, we consider partial distinguishability noise, one of the major noise sources that hinders experimental demonstration of quantum computational advantage. 
We adopt the uniform indistinguishability rate $x$ as a noise model, which indicates that each photon becomes distinguishable with probability $1 - x$~\cite{tichy2015sampling, renema2018efficient, moylett2019classically}.
Our key ingredient for hardness result is that, for a fixed $x^{*}$ of the partially distinguishable boson sampling, we consider the indistinguishability rate $x \in [0,x^{*}]$ as an input for that noisy boson sampler, and derive the hardness result for that sampling problem. 
Then, under the plausible assumption that classical simulability does not degrade with decreasing $x$, this implies the classical hardness of simulating partially distinguishable boson sampling with a fixed indistinguishability rate $x^{*}$.

To summarize our result, we show that noisy $N$-single photon boson sampling that has $(1-x^{*})N = O(\log N)$ number of distinguishable photons on average maintains the classical intractability of the ideal boson sampling.
Our main strategy is the \textit{low-degree polynomial approximation} of noisy output probability, which is also a key technique for classical simulation studies for noisy boson sampling. 
We also generalize our result to the photon loss noise, and show that noisy boson sampling that unavoidably has on average $O(\log N)$ number of lost photons and $O(\log N)$ number of distinguishable photons maintains the classical intractability of the ideal boson sampling.  


\subsection{Classical hardness results for noisy boson sampling}

Hereafter, we compare our results to existing literature on the complexity of noisy boson sampling.  
One of the closely related works to our work is Ref.~\cite{shchesnovich2014sufficient}, which shows that partially distinguishable boson sampling, whose single-photon mode mismatch is given by $O(N^{-3/2})$, still preserves the classical intractability of the ideal boson sampling. 
Despite a slight difference in definition, the single-photon mode mismatch in Ref.~\cite{shchesnovich2014sufficient} represents a single-photon distinguishability, which is essentially the same quantity as $1-x$ in our work; this implies the classical hardness of partially distinguishable boson sampling for $O(N^{-1/2})$ number of distinguishable photons on average. 
Also, Ref.~\cite{aaronson2016bosonsampling} shows that lossy boson sampling, where at most a fixed $k = O(1)$ number of photons are lost over $N$ input photons, maintains the same complexity as the ideal boson sampling case.

The main strategy for these previous works is to show that the output probability estimation problem for the ideal case (i.e., $|\text{GPE}|_{\pm}^{2}$ problem in the main text) can also be solved by a noisy boson sampler under Stockmeyer's reduction, when the noise rate is below a certain threshold. 
In our work, we similarly adopt this strategy while improving the tolerable noise threshold for the classical intractability of noisy boson sampling; specifically, we increase the number of tolerable noisy photons for classical hardness to $O(\log N)$, now allowing the number of noisy photons to \textit{scale} with the system size.
Hence, by providing a significantly improved noise threshold for the classical hardness of boson sampling, our work paves the way toward realizing quantum computational advantage with noisy boson samplers.



}

\cor{

\subsection{Classical simulability results for partially distinguishable boson sampling}

Another closely related line of work is the classical simulability results of noisy boson sampling with partially distinguishable photons.
First of all, Ref.~\cite{kalai2014gaussian} finds that when the rate of Gaussian-type noise is sufficiently large, the output probability of that noisy boson sampling can be efficiently approximated by using a \textit{low-degree polynomial approximation}. 
Reference~\cite{renema2018efficient} generalizes this low-degree polynomial approximation technique in Ref.~\cite{kalai2014gaussian} to 
partially distinguishable boson sampling with uniform indistinguishability rate $x$, and proposes a classical algorithm that approximately simulates partially distinguishable boson sampling.
References~\cite{moylett2019classically, renema2018classical} further generalize to the combined noise model for loss and distinguishability, characterized by uniform transmission rate $\eta$ and uniform indistinguishability $x$, and similarly employing the low-degree polynomial approximation, they propose a classical algorithm to approximately simulate that noisy boson sampling. 
Reference~\cite{shchesnovich2019noise} also generalizes Ref.~\cite{kalai2014gaussian} by converting the Gaussian noise to realistic partial distinguishability noise, and suggests a noise threshold for efficient classical simulability of that noisy boson sampling. 
References~\cite{renema2020simulability, shi2022effect} extend the analysis for the Gaussian boson sampling and also propose an efficient classical simulation for sufficiently small indistinguishability rates. 
Reference~\cite{van2024efficient} generalizes these previous classical algorithms from uniform distinguishability noise to the non-uniform distinguishability cases.

These numerous simulability results suggest a noise threshold for the classical simulability of partially distinguishable boson sampling, providing a necessary condition for achieving quantum computational advantage. 
Here, the major technique they employed is a low-degree polynomial approximation, which can be achieved by truncation of high-order terms in the noisy output probability.
In our work, we employ a similar low-degree polynomial approximation strategy for noisy output probabilities, but using it to establish the classical hardness result for partially distinguishable boson sampling.
Our result provides a noise threshold for the classical intractability of partially distinguishable boson sampling, thereby complementing existing simulability results and ultimately offering a sufficient condition for the noise rate of boson sampling to achieve the quantum computational advantage. 

}

\cor{

\subsection{Classical simulability results for noisy boson sampling with photon loss}

We list here the classical simulability results for noisy boson sampling under photon loss.
Reference~\cite{rahimi2016sufficient} presents a classical simulation method based on quasi-probability distribution of continuous-variable quantum states, showing that boson sampling becomes classically simulable when the loss and dark count rate are sufficiently large. 
After that, Refs.~\cite{oszmaniec2018classical, garcia2019simulating} show that for uniform photon loss, the lossy boson sampling becomes efficiently simulable when $O(\sqrt{N})$ photons survive over $N$ input photons. 
Reference~\cite{qi2020regimes} reproduces this scaling behavior for the Gaussian boson sampling case, showing the classical simulability of lossy Gaussian boson sampling with uniform loss when $O(\sqrt{N})$ photons survive over $N$ input photons on average. 
Reference~\cite{brod2020classical} generalizes these previous analyses to the non-uniform loss case. 
For the non-asymptotic case, by using the effect of photon loss on Gaussian boson sampling, Refs.~\cite{villalonga2021efficient, bulmer2022boundary, oh2024classicalalgorithm} develop classical algorithms that can simulate the recent Gaussian boson sampling experiments~\cite{zhong2020quantum, zhong2021phase, madsen2022quantum, deng2023gaussian}.

Although not directly related to our work, these works illustrate how noise in experiments can significantly reduce the computational complexity of boson sampling, thereby hindering the realization of quantum computational advantage. 
Hence, these simulability results highlight the necessity for further classical hardness results of noisy boson sampling, i.e., improving the tolerable noise rates that maintain the computational complexity of boson sampling. 
}


\cor{

\subsection{Classical simulability and hardness results for shallow-depth boson sampling}

We now provide an overview of the shallow-depth boson sampling.
For the classical simulability results, Refs.~\cite{vidal2003efficient, garcia2019simulating} present classical algorithms for simulating $1$-d log-depth boson sampling based on MPS methods, for geometrically local circuit architectures. 
Reference~\cite{qi2022efficient} extends these arguments, showing that $1$-d local log-depth Gaussian boson sampling can be efficiently simulated. 
Classical algorithms for shallow-depth boson sampling with general-dimensional local circuit architectures have also been developed, including boson sampling in Refs.~\cite{deshpande2018dynamical, oh2022classical}, and Gaussian boson sampling in Refs.~\cite{oh2022classical, kolarovszki2023simulating}.
Reference~\cite{oh2024classical} proposes a classical algorithm for noisy constant-depth boson sampling with general types of noise, including partial distinguishability noise that we mainly consider.

Meanwhile, for the classical hardness results of shallow-depth boson sampling, Ref.~\cite{go2024exploring} proposes a log-depth circuit architecture with geometrically non-local gates, that can potentially be used for hardness results. 
Subsequently, Ref.~\cite{go2024computational} claims the average-case hardness for the log-depth circuit architecture proposed in Ref.~\cite{go2024exploring} up to a certain imprecision level, providing hardness evidence of shallow-depth boson sampling for that circuit architecture.

Indeed, analysis on shallow-depth boson sampling is crucial for achieving quantum computational advantage, because shallow-depth circuits can significantly reduce the effect of noise.
However, since the hardness results on noisy boson sampling (including ours) consider Haar-random linear optical circuits that require a deep circuit depth~\cite{clements2016optimal, russell2017direct}, they cannot be directly applied to those analyses of shallow-depth boson sampling. 
Hence, extending the existing hardness arguments for noisy boson sampling to shallow-depth circuit architectures would pave the way toward the ultimate demonstration of quantum computational advantage.

}

\section{Proof of Lemma~\ref{lemma:distancebytruncationissmall}}\label{proofoflemma}\label{proofoflemma}

In this appendix we prove Lemma~\ref{lemma:distancebytruncationissmall}, which states that for $X \sim \mathcal{N}(0,1)_{\mathbb{C}}^{N\times N}$ and $x \ge 1 - \frac{k_{\text{max}}}{N}$, the truncated output probability $q^{(l)}(x,X)$ in Eq.~\eqref{truncatedoutputprobability} is $\epsilon_1$-close to $q(x, X)$ in Eq.~\eqref{noisyoutputprobability}, for $\epsilon_1$ given by
\begin{equation}\label{epsilon1appendix}
    \epsilon_1 = \delta^{-1}N^{-\frac{c_{\text{max}}}{3}\left(\frac{c_l}{c_{\text{max}}} - 1 \right)^2 },
\end{equation}
over at least $1-\delta$ of $X \sim \mathcal{N}(0,1)_{\mathbb{C}}^{N\times N}$.

\begin{proof}[Proof of Lemma~\ref{lemma:distancebytruncationissmall}]

The error induced by truncation from $q(x,X)$ to $q^{(l)}(x,X)$ can be expressed as
\begin{equation}\label{errorterm}
q(x,X) - q^{(l)}(x, X) = \sum_{j=0}^{N-l-1}\binom{N}{j}x^{j}(1-x)^{N-j}  q_{j}(X) .
\end{equation}
Here, given $l > k$, the coefficients in Eq.~\eqref{errorterm} correspond to the tail of the binomial distribution. 
We first investigate the maximum size of these binomial tails in terms of $l$ and $k$. 

To examine the bound of binomial tails, we use the Chernoff-Hoeffding bound. 
Specifically, for $N$ number of independent Bernoulli random variables $X_1,\dots,X_N$ with each $X_i \in \{0,1\}$, the sum of these variables $S_N=\sum_{i=1}^N X_i$ satisfies the following inequality~\cite{dubhashi2009concentration}: 
\begin{align}\label{chernoffhoeffding}
    \Pr[S_N\geq (1+\xi)\mathbb{E}[S_N]]&\leq \exp\left(-\frac{\xi^2}{3}\mathbb{E}[S_N]\right).
\end{align}
Using this bound we can find the upper bound of the following tail terms:
\begin{align}\label{eieieis}
    \sum_{j=0}^{N-l-1} \binom{N}{j}x^{j}(1-x)^{N-j}
    =\sum_{j=l+1}^N \binom{N}{j}(1-x)^{j} x^{N-j}.
\end{align}
Specifically, by considering the success probability of the Bernoulli trial (i.e., probability that $X_i = 1$) as $1 - x$, the right-hand side of Eq.~\eqref{eieieis} corresponds to $\Pr[S_N \geq l+1]$. 
Hence, by using $\mathbb{E}[S_N] = (1-x)N = k$ and setting $\xi=(l+1)/\mathbb{E}[S_N]-1=(l+1)/k-1$, we have 
\begin{align}
    &\sum_{j=0}^{N-l-1} \binom{N}{j}x^{j}(1-x)^{N-j} \nonumber \\
    &=\Pr[S_N\geq l+1] \\
    &\leq \exp\left(-\frac{k}{3}\left(\frac{l+1}{k}-1\right)^2\right) \\ 
    &\leq \exp\left(-\frac{k}{3}\left(\frac{l}{k}-1\right)^2\right) \\
    &\leq \exp\left(-\frac{k_{\text{max}}}{3}\left(\frac{l}{k_{\text{max}}}-1\right)^2\right) \\
    &= N^{-\frac{c_{\text{max}}}{3}\left(\frac{c_l}{c_{\text{max}}} - 1 \right)^2 } .
\end{align}

We now examine the size of $q_{j}(X)$ in Eq.~\eqref{errorterm}.
For any $j = 0,\dots,N$, the average value of $q_{j}(X)$ in Eq.~\eqref{fixedpdoutputprobability} over $X \sim \mathcal{N}(0,1)_{\mathbb{C}}^{N\times N}$ is
\begin{align}
    &\E_{X} \left[ q_{j}(X) \right] \nonumber \\
    &=  \binom{N}{j}^{-1}\sum_{\substack{I \subseteq [N]\\|I|=j}} \sum_{\substack{J \subseteq [N]\\|J|=j}} \E_{X} \left[ |\Per(X_{I,J})|^2 \Per(|X_{\bar{I},\bar{J}}|^2) \right] \\    
    &=  \binom{N}{j}^{-1}\sum_{\substack{I \subseteq [N]\\|I|=j}} \sum_{\substack{J \subseteq [N]\\|J|=j}} j! (N-j)! \label{averageofpartiallydistinguishablepermanent}\\
    &= N! ,
\end{align}
where the equality in Eq.~\eqref{averageofpartiallydistinguishablepermanent} comes from the fact that for any $I, J \subseteq [N]$ with $|I|=|J|=j$, $X_{I,J} \sim \mathcal{N}(0,1)_{\mathbb{C}}^{j\times j}$ and $X_{\bar{I},\bar{J}} \sim \mathcal{N}(0,1)_{\mathbb{C}}^{N-j\times N-j}$ are i.i.d. Gaussians independent each other.
Also, for $X \sim \mathcal{N}(0,1)_{\mathbb{C}}^{N\times N}$, one can obtain $\E_{X}[|\Per(X)|^2] = \E_{X}[\Per(|X|^2)] = N!$ as shown in Ref.~\cite{aaronson2011computational}. 
By the above equation and using the fact that $q(x, X) \ge q^{(l)}(x,X)$, one can find that
\begin{align}
    &\E_{X}\left[ |q(x, X) -q^{(l)}(x,X)| \right] \nonumber \\
    &=  \sum_{j=0}^{N-l-1}\binom{N}{j}x^{j}(1-x)^{N-j}  \E_{X} \left[ q_{j}(X) \right] \\
    &= N!\sum_{j=0}^{N-l-1}\binom{N}{j}x^{j}(1-x)^{N-j} .
\end{align}
Finally, by using Markov's inequality, given the size of $\epsilon_1$ as in Eq.~\eqref{epsilon1appendix}, observe that
\begin{align}
    &\Pr_{X}\left[|q(x, X) - q^{(l)}(x,X)| > \epsilon_1 N! \right] \nonumber  \\
    &\le \frac{1}{\epsilon_1 N!}\E_{X}\left[ |q(x, X) - q^{(l)}(x,X)| \right] \\
    & = \frac{1}{\epsilon_1}\sum_{j=0}^{N-l-1}\binom{N}{j}x^{j}(1-x)^{N-j} \\
    &\leq \delta ,
\end{align}
concluding the proof.
\end{proof}

\section{Specifying $c_{\text{max}}$ and $c_l$}\label{determiningcmaxandcl}

In this appendix, we investigate the appropriate size of the parameters $c_{\text{max}}$ and $c_l$, first to make $\epsilon_1 = \delta^{-1}N^{-\frac{c_{\text{max}}}{3}\left(\frac{c_l}{c_{\text{max}}} - 1 \right)^2 } \ll \epsilon_0N^{-c_{\text{max}}-c_l(1+\log\Delta^{-1})}$, while avoiding $\epsilon_0N^{-c_{\text{max}}-c_l(1+\log\Delta^{-1})}$ becomes exceedingly small. 
To do so, we aim to find a minimum $c_l$ that satisfies the first condition. 

More specifically, to satisfy the first condition, the following quantity should be small, to at most a small constant:
\begin{align}
&\frac{\epsilon_1}{\epsilon_0N^{-c_{\text{max}}-c_l(1+\log\Delta^{-1})}} \nonumber \\
&= \delta^{-1}\epsilon_{0}^{-1}N^{-\frac{c_{\text{max}}}{3}\left(\frac{c_l}{c_{\text{max}}} - 1 \right)^2 + c_{\text{max}} + c_l(1+\log\Delta^{-1})} \\
&= O\left( l\delta_{0}^{-1}\epsilon_{0}^{-1} N^{-\frac{c_{\text{max}}}{3}\left(\frac{c_l}{c_{\text{max}}} - 1 \right)^2 + c_{\text{max}} + c_l(1+\log\Delta^{-1})}\right)  \label{b3} \\
&= O\left( l N^{-\frac{c_{\text{max}}}{3}\left(\frac{c_l}{c_{\text{max}}} - 1 \right)^2 + c_{\text{max}} + c_l(1+\log\Delta^{-1})  +  \frac{\log\epsilon_0^{-1}\delta_0^{-1}}{\log N}}\right) , \label{b4}
\end{align}
where we used $\delta^{-1} = O(\delta_0^{-1}l)$ in Eq.~\eqref{b3}, given that $\delta_0 = 1 - (1-2\delta)^{l+1}$.

Our approach is to first to parameterize $c_l$ in terms of $c_{\text{max}}$ as $c_l = \alpha c_{\text{max}}$ for $\alpha > 1$.
Then the exponent in Eq.~\eqref{b4} can be expressed as a quadratic function of $\alpha$ as
\begin{align}
-\frac{c_{\text{max}}}{3}\left( \alpha^2 -(5+3\log\Delta^{-1}) \alpha - 2 - \frac{3\log\epsilon_0^{-1}\delta_0^{-1}}{c_{\text{max}}\log N}  \right)   .  
\end{align}
We denote $\alpha^*$ as a (positive) root of the above quadratic function, which is 
\begin{align}\label{b6}
    \alpha^* &= \frac{5+3\log\Delta^{-1}}{2} \nonumber \\
    &+ \sqrt{\left(\frac{5+3\log\Delta^{-1}}{2}\right)^2 + 2 + \frac{3\log\epsilon_0^{-1}\delta_0^{-1}}{c_{\text{max}}\log N}} .
\end{align}
To prevent the right-hand side of Eq.~\eqref{b4} increase with system size $N$, the exponent in Eq.~\eqref{b4} should be negative, and thus $\alpha$ should be larger than $\alpha^*$.
Here, we set $\alpha$ as  
\begin{align}\label{b7}
    \alpha = \lambda \frac{5+3\log\Delta^{-1}}{2} + \frac{3\log\epsilon_0^{-1}\delta_0^{-1}}{c_{\text{max}}\log N} 
\end{align}
for a constant $\lambda$. 
One can check that when the constant $\lambda$ is larger than a certain value, $\alpha > \alpha^*$ always holds, by using the following inequalities:
\begin{align}
&\frac{5+3\log\Delta^{-1}}{2} + \sqrt{\left(\frac{5+3\log\Delta^{-1}}{2}\right)^2 + 2 + \frac{3\log\epsilon_0^{-1}\delta_0^{-1}}{c_{\text{max}}\log N}} \\
&\le \frac{5+3\log\Delta^{-1}}{2} + \sqrt{\left(\frac{5+3\log\Delta^{-1}}{2}\right)^2 + 2}  + \frac{3\log\epsilon_0^{-1}\delta_0^{-1}}{c_{\text{max}}\log N} \\
&< \lambda \frac{5+3\log\Delta^{-1}}{2} + \frac{3\log\epsilon_0^{-1}\delta_0^{-1}}{c_{\text{max}}\log N} , \label{b9}
\end{align}
where the inequality in Eq.~\eqref{b9} holds for $\lambda \ge 2.149$ given that $\frac{5+3\log\Delta^{-1}}{2} > 2.5$ (as $\log\Delta^{-1} = \log\frac{c_{\text{max}}+c_{\text{min}}}{c_{\text{max}}-c_{\text{min}}} > 0$).
Since we aim to minimize $\alpha$ (to minimize $c_l$), we set $\alpha$ as given in Eq.~\eqref{b7} with $\lambda = 2.149$ hereafter. 

We now examine the scaling behavior of the right-hand side term in Eq.~\eqref{b4}, for $\alpha$ we have set. 
More specifically, we find that 
\begin{align}
&N^{-\frac{c_{\text{max}}}{3}\left( \alpha^2 -(5+3\log\Delta^{-1}) \alpha - 2 - \frac{3\log\epsilon_0^{-1}\delta_0^{-1}}{c_{\text{max}}\log N}  \right)} \\
&= N^{-\frac{c_{\text{max}}}{3}\left( (2\alpha^* - 5-3\log\Delta^{-1})(\alpha - \alpha^*) + (\alpha - \alpha^*)^2  \right)} \label{b11} \\
&\leq N^{-\frac{c_{\text{max}}}{3}\left( (2\alpha^* - 5-3\log\Delta^{-1})(\alpha - \alpha^*) \right)} \\
&=  N^{-\frac{2c_{\text{max}}}{3}(\alpha - \alpha^*)\sqrt{\left(\frac{5+3\log\Delta^{-1}}{2}\right)^2 + 2 + \frac{3\log\epsilon_0^{-1}\delta_0^{-1}}{c_{\text{max}}\log N}} } , \label{b13}
\end{align}
where we used Taylor expansion of the quadratic function in Eq.~\eqref{b11}.
Here, from the definition of $\alpha$ and $\alpha^*$, one can check that $\alpha - \alpha^* = \Omega(1 + \frac{\log\epsilon_0^{-1}\delta_0^{-1}}{\log N})$.
Then the right-hand side of Eq.~\eqref{b13} scales inverse-polynomially with $N$ and $\epsilon_0^{-1}\delta_0^{-1}$. 
Also, we have 
\begin{align}
l &= c_l\log N \\ 
&= \alpha c_{\text{max}}\log N \\
&= O(\log N + \log \epsilon_0^{-1}\delta_0^{-1}),
\end{align}
such that $l$ scales logarithmically with $N$ and $\epsilon_0^{-1}\delta_0^{-1}$.
Combining these results, the right-hand side term in Eq.~\eqref{b4} scales inverse-polynomially with $N$ and $\epsilon_0^{-1}\delta_0^{-1}$ up to logarithmic factor; thus, from Eq.~\eqref{b4} we have 
\begin{align}\label{b17}
\frac{\epsilon_1}{\epsilon_0N^{-c_{\text{max}}-c_l(1+\log\Delta^{-1})}} = O(\text{poly}(N,\epsilon_0^{-1}\delta_0^{-1})^{-1}).  
\end{align} 

Next, we parameterize $c_{\text{max}}$ in terms of $c_{\text{min}}$ as $c_{\text{max}} = \kappa c_{\text{min}}$ for a constant $\kappa$. 
Then, $\Delta = \frac{c_{\text{max}}-c_{\text{min}}}{c_{\text{max}}+c_{\text{min}}} = \frac{\kappa - 1}{\kappa + 1}$, and we have
\begin{align}
&\epsilon_0N^{-c_{\text{max}}-c_l(1+\log\Delta^{-1})} \nonumber \\
&= \epsilon_0N^{-(1+ (1+\log\Delta^{-1})\alpha)c_{\text{max}}} \\
&= \epsilon_0N^{-\left(1+ (1+\log\Delta^{-1})\left(\lambda \frac{5+3\log\Delta^{-1}}{2} + \frac{3\log\epsilon_0^{-1}\delta_0^{-1}}{c_{\text{max}}\log N}\right) \right) c_{\text{max}}} \\
&= \epsilon_0^{4 +3\log\frac{\kappa + 1}{\kappa - 1}}\delta_0^{3+3\log\frac{\kappa + 1}{\kappa - 1}}  \nonumber \\
&\times N^{-\kappa\left(1+ \frac{\lambda}{2}\left(1+\log\frac{\kappa + 1}{\kappa - 1}\right)\left(5+3\log\frac{\kappa + 1}{\kappa - 1}\right) \right) c_{\text{min}}} . \label{b20}
\end{align}
Since we have freedom to choose the constant $\kappa$, we choose $\kappa = 2.108$, which maximizes the exponent of $N$ in Eq.~\eqref{b20} to $-39.35 c_{\text{min}}$.

To sum up, by setting $c_{\text{max}}$ and $c_l$ as 
\begin{align}
    &c_{\text{max}} = \kappa c_{\text{min}} , \label{cmax} \\
    &c_l =  \kappa \lambda \frac{5+3\log\frac{\kappa + 1}{\kappa - 1}}{2} c_{\text{min}} + \frac{3\log\epsilon_0^{-1}\delta_0^{-1}}{\log N} 
\end{align}
with $\kappa = 2.108$ and $\lambda = 2.149$, we have 
\begin{align}
     \epsilon_0N^{-c_{\text{max}}-c_l(1+\log\Delta^{-1})} - \epsilon_1 = O\left( \epsilon_0N^{-c_{\text{max}}-c_l(1+\log\Delta^{-1})} \right) , \label{b23}
\end{align}
and the exponents in the right-hand side of Eq.~\eqref{b23} can be specified as
\begin{align}\label{b25}
\epsilon_0N^{-c_{\text{max}}-c_l(1+\log\Delta^{-1})} = \epsilon_0^{7.094}\delta_0^{6.094} N^{-39.35 c_{\text{min}}} .
\end{align}

Note that one can further optimize Eq.~\eqref{b25}, by choosing $\alpha$ in Eq.~\eqref{b7} more closer to $\alpha^*$ in Eq.~\eqref{b6} while satisfying Eq.~\eqref{b17}.
But, for our purpose, this result is sufficient.

\section{Proof of Corollary~\ref{sublogpdhardness}: Reducing $|\text{GPE}|_{\pm}^{2}$ to $|\text{PGPE}|_{\pm}^{2}$ for $k_{\text{min}} = o(\log N)$}\label{appendix:proofofcorollary}

To prove Corollary~\ref{sublogpdhardness}, we reproduce Lemma~\ref{lemma:distancebytruncationissmall} for sub logarithmically scaling $k_{\text{min}} = o(\log N)$.
Similarly to our main result in Sec.~\ref{Section:reduction}, let $x_{\text{min}}$ be a minimum value of $x$ we set, such that $x_{\min} = 1 - \frac{k_{\text{max}}}{N}$ with $k_{\text{max}} = o(\log N)$.
Specifically, $k_{\text{max}}$ is larger than $k_{\text{min}}$, and also proportional to $k_{\text{min}}$ such that $\frac{k_{\text{max}}}{k_{\text{min}}} = \Theta(1)$.

We first examine the size of the binomial term in the noisy output probability $q(x,X)$ in Eq.~\eqref{noisyoutputprobability}, for $j = N - l$ with $l$ larger than $k_{\text{max}}$. 
Specifically, for any input indistinguishability rate $x = 1-\frac{k}{N}$ with $k \in [k_{\text{min}}, k_{\text{max}}]$, the binomial term corresponding to $j = N-l$ is bounded by
\begin{align}
    &\binom{N}{N-l}x^{N-l}(1-x)^{l} \nonumber \\
    &= \binom{N}{N-l}\left(1-\frac{k}{N}\right)^{N-l}\left(\frac{k}{N}\right)^{l} \\
    &< \frac{1}{\sqrt{l}} \frac{N^N}{l^l(N-l)^{N-1}}\left(1-\frac{k}{N}\right)^{N-l}\left(\frac{k}{N}\right)^{l} \label{asxxz}\\
    &= \frac{1}{\sqrt{l}} \left(\frac{N-k}{N-l}\right)^{N-l}\left(\frac{k}{l}\right)^{l}\\
    &< \frac{1}{\sqrt{l}} e^{-l\log(\frac{l}{k}) + l - k } \label{14} \\
    &< e^{-l\log(\frac{l}{k}) + l - k } \\
    &\le  e^{-l\log(\frac{l}{k_{\text{max}}}) + l - k_{\text{max}} }, \label{15} 
\end{align}
where we used Stirling's inequality in Eq.~\eqref{asxxz}, and used the following relation in Eq.~\eqref{14}:
\begin{align}
    &(N-l)\log(\frac{N-k}{N-l}) \nonumber \\
    &= (N-l)\left(\log(1-\frac{k}{N})-\log(1-\frac{l}{N})\right) \\ 
    &= (N-l)\sum_{m=1}^{\infty} \frac{1}{m!}\frac{l^m - k^m}{N^m} \label{18}\\
    &= l - k + \sum_{m=1}^{\infty} \frac{l^{m+1}-k^{m+1}-(m+1)l(l^{m}-k^m)}{(m+1)!N^m} \\
    &\le l - k + \sum_{m=1}^{\infty} \frac{-l^{m+1} + 2lk^m - k^{m+1}}{(m+1)!N^m} \label{20} \\ 
    &= l - k - \sum_{m=1}^{\infty} \frac{k^{m+1}}{(m+1)!N^m}\left(\left(\frac{l}{k}\right)^{m+1} - 2\left(\frac{l}{k}\right) + 1 \right) \\ 
    &< l - k , \label{22}
\end{align}
where we used Taylor expansion in Eq.~\eqref{18}, and used $m \ge 1$ and $l > k$ in Eq.~\eqref{20} and Eq.~\eqref{22}.

By the property of the binomial distribution, the binomial term $\binom{N}{j}x^{j}(1-x)^{N-j}$ monotonically decreases with decreasing $j$, as long as $j < Nx$.
Hence, for any $x = 1 - \frac{k}{N}$, all the binomial terms corresponding to $j < N - l$ are smaller than the binomial term corresponding to $j = N - l$, given that $l$ is larger than $k_{\text{max}}$.
Accordingly, the summation of the binomial terms from $j=0$ to $j=N-l-1$ is upper-bounded as
\begin{align}
    &\sum_{j=0}^{N-l-1}\binom{N}{j}x^{j}(1-x)^{N-j} \nonumber \\
    &\le (N-l)\cdot\binom{N}{N-l}x^{N-l}(1-x)^{l} \\
    &< Ne^{-l\log(\frac{l}{k_{\text{max}}}) + l - k_{\text{max}} } .\label{upperboundofsumofbinomials}
\end{align}

Then, similarly using the proof of Lemma~\ref{lemma:distancebytruncationissmall}, $q(x, X)$ and $q^{(l)}(x,X)$ are $\epsilon_1N!$-close over at least $1 - \delta$ of X, for $\epsilon_1 = \delta^{-1}Ne^{-l\log(\frac{l}{k_{\text{max}}}) + l - k_{\text{max}} }$, i.e., 
\begin{align}
\begin{split}
    \Pr_{X}\left[|q(x, X) - q^{(l)}(x,X)| > \epsilon_1 N! \right] < \delta .
\end{split}
\end{align}
By using this property and following the proof in Sec.~\ref{proofoftheorem1}, the proof of Corollary~\ref{sublogpdhardness} is now straightforward. 

\begin{proof}[Proof of Corollary~\ref{sublogpdhardness}]
The proof is essentially the same as the proof of Theorem~\ref{logpdhardness} except for some minor changes in the size of the parameters. 
Let $\mathcal{O}$ be the oracle that on input $x$ and $X \sim \mathcal{N}(0,1)_{\mathbb{C}}^{N\times N}$, estimates $q(x, X)$ within $\epsilon N!$ over $1-\delta$ of $X$, such that 
\begin{equation}
    \Pr_{X}\left[|\mathcal{O}(x,X) - q(x, X)| > \epsilon N! \right] < \delta.
\end{equation}
For $\epsilon' = \epsilon + \epsilon_1$, using triangular inequality, 
\begin{align}
    &\Pr_{X}\left[|\mathcal{O}(x,X) - q^{(l)}(x,X)| > \epsilon'N! \right]  \nonumber\\ 
    &\le \Pr_{X}\left[|\mathcal{O}(x,X) - q(x, X)| > \epsilon N! \right] \nonumber \\
    &+ \Pr_{X}\left[|q(x, X) - q^{(l)}(x,X)| > \epsilon_1 N! \right] \\
    &< 2\delta.
\end{align}

Let $\epsilon'' = \epsilon'e^{k_{\text{max}}}$ such that $\epsilon'x^{-N+l} \le \epsilon'x_{\text{min}}^{-N+l} = \epsilon'e^{k_{\text{max}} - \tilde{\mathcal{O}}(N^{-1})} \le \epsilon''$. Then, for $f_{X}^{(l)}(x) = x^{-N + l}q^{(l)}(x,X)$, 
\begin{align}
    &\Pr_{X}\left[|\mathcal{O}(x,X)x^{-N+l} - f_{X}^{(l)}(x)| > \epsilon''N! \right] \nonumber \\
    &\le \Pr_{X}\left[|\mathcal{O}(x,X)x^{-N+l} - f_{X}^{(l)}(x)| > \epsilon'x^{-N+l}N! \right] \\
    &= \Pr_{X}\left[|\mathcal{O}(x,X) - q^{(l)}(x,X)| > \epsilon'N! \right] \\
    &< 2\delta .\label{estimatelowdegreepolynomialforsublogpd}
\end{align}
As Eq.~\eqref{estimatelowdegreepolynomialforsublogpd} is the same as Eq.~\eqref{estimatelowdegreepolynomial} in the proof of Theorem~\ref{logpdhardness}, we can repeat all the steps identically to the proof of Theorem~\ref{logpdhardness} and obtain the estimation value $|\Per(X)|^2$ within $\epsilon''e^{l(1+\log\Delta^{-1})}N!$ with probability at least $1 - (1-2\delta)^{l+1}$ (where $\Delta = \frac{k_{\text{max}} - k_{\text{min}}}{k_{\text{max}} + k_{\text{min}}} = \Theta(1)$).
Then, the conditions for the error parameters $\epsilon$ and $\delta$ are $\delta_0 \ge 1 - (1-2\delta)^{l+1}$ and $\epsilon_0 \ge \epsilon''e^{l(1+\log\Delta^{-1})} = (\epsilon + \epsilon_1)e^{k_{\text{max}}}e^{l(1+\log\Delta^{-1})}$ (in other words, $\epsilon \le \epsilon_0e^{-k_{\text{max}} - l(1+\log\Delta^{-1})} - \epsilon_1$. 
Hence, similarly as before, we set $\delta_0 = 1 - (1-2\delta)^{l+1}$, such that $\delta = O(l^{-1}\delta_0)$.

Here, note that the approximation error $\epsilon$ in the $|\text{PGPE}|_{\pm}^{2}$ problem is positive.
Therefore, as we have previously discussed in Sec.~\ref{proofoftheorem1}, $\epsilon_1$ should satisfy $\epsilon_1 \ll \epsilon_0e^{-k_{\text{max}}-l(1+\log\Delta^{-1})}$ such that $\epsilon_0e^{-k_{\text{max}} - l(1+\log\Delta^{-1})} - \epsilon_1 = O(\epsilon_0e^{-k_{\text{max}} - l(1+\log\Delta^{-1})})$.
This means that the following quantity should be small enough: 
\begin{align}
    &\frac{\epsilon_1}{\epsilon_0e^{-k_{\text{max}}-l(1+\log\Delta^{-1})}} \nonumber \\
    &= \epsilon_0^{-1}\delta^{-1}Ne^{-l\log(\frac{l}{k_{\text{max}}}) + l(2+\log\Delta^{-1}) } \\
    &= O\left( l e^{-l\log(\frac{l}{k_{\text{max}}}) + l(2+\log\Delta^{-1}) +\log N + \log \epsilon_0^{-1}\delta_0^{-1}} \right), \label{c22}
\end{align}
where we used $\delta^{-1} = O(l\delta^{-1})$ from $\delta_0 = 1 - (1-2\delta)^{l+1}$ in Eq.~\eqref{c22}. 
To make the right-hand side of Eq.~\eqref{c22} small, we set logarithmically scaling $l$ with $N$ and $\epsilon_0^{-1}\delta_0^{-1}$, such that  
\begin{align}\label{c23}
    l = a_0\log N + a_1\log\epsilon_0^{-1}\delta_0^{-1}, 
\end{align}
for constants $a_0, a_1 > 0$. 
Note that the right-hand side of Eq.~\eqref{c22} scales $o(\text{poly}(N, \epsilon_0^{-1}\delta_0^{-1})^{-1})$ even for arbitrarily small constants $a_0$ and $a_1$, because the leading term scales $e^{-l\log l}$. 
Hence, for $l$ given in Eq.~\eqref{c23} with any $a_0, a_1 > 0$, $\epsilon_0e^{-k_{\text{max}} - l(1+\log\Delta^{-1})} - \epsilon_1 = O(\epsilon_0e^{-k_{\text{max}} - l(1+\log\Delta^{-1})})$.

Therefore, using $l$ in Eq.~\eqref{c23}, $\delta = O(l^{-1}\delta_0) = O(\delta_0(\log N + \log\epsilon_0^{-1}\delta_0^{-1})^{-1})$.
Also, the condition for $\epsilon$ reduces to 
\begin{align}
\epsilon &\le O\left(\epsilon_0e^{-k_{\text{max}} - l(1+\log\Delta^{-1})} \right) \\
&= O\left(\frac{\epsilon_0^{1+a_1(1+\log\Delta^{-1})}\delta_0^{a_1(1+\log\Delta^{-1})}}{N^{a_0(1+\log\Delta^{-1})}}e^{-k_{\text{max}}} \right) . \label{c25}
\end{align}
Because $a_0$ and $a_1$ can be made arbitrarily small constant, setting $\epsilon = O(\epsilon_0^{\alpha+1}\delta_0^{\alpha}N^{-\alpha})$ for any constant $\alpha >0$ is sufficient to satisfy the condition in Eq.~\eqref{c25}. 
This completes the proof.

\end{proof}

\end{document}